\documentclass[11pt,a4paper]{article}

\usepackage[utf8]{inputenc}
\usepackage[T1]{fontenc}
\usepackage[english]{babel}
\usepackage{amsmath,amssymb,amsthm,mathtools}
\usepackage{mathrsfs}
\usepackage{geometry}
\usepackage{booktabs}
\usepackage{enumitem}
\usepackage[final,expansion=false]{microtype}
\usepackage{xcolor}
\usepackage{caption}
\usepackage{authblk}
\usepackage[hidelinks]{hyperref}

\theoremstyle{plain}
\newtheorem{theorem}{Theorem}[section]
\newtheorem{lemma}[theorem]{Lemma}
\newtheorem{proposition}[theorem]{Proposition}
\newtheorem{corollary}[theorem]{Corollary}
\theoremstyle{definition}
\newtheorem{definition}[theorem]{Definition}
\newtheorem{remark}[theorem]{Remark}

\newtheorem{assumption}[theorem]{Assumption}

\newcommand{\R}{\mathbb{R}}
\newcommand{\C}{\mathbb{C}}
\newcommand{\Z}{\mathbb{Z}}
\newcommand{\eco}{\mathcal{E}}
\newcommand{\deq}{\vcentcolon=}
\newcommand{\re}{\operatorname{Re}}
\newcommand{\im}{\operatorname{Im}}
\newcommand{\ind}{\operatorname{ind}}
\newcommand{\sgn}{\operatorname{sgn}}
\newcommand{\diag}{\operatorname{diag}}

\newcommand{\Sym}{\operatorname{Sym}}
\newcommand{\supp}{\operatorname{supp}}
\newcommand{\norm}[1]{\left\lVert #1 \right\rVert}
\newcommand{\abs}[1]{\left\lvert #1 \right\rvert}
\newcommand{\ph}{\widehat{\varphi}}
\newcommand{\one}{\mathbf{1}}
\title{\bfseries Monodromy of Walrasian Equilibrium Prices:\\
A Fixed-Preference Endowment-Redistribution Construction}
\author[1]{Andrea Loi}
\author[2]{Stefano Matta}
\affil[1]{Dipartimento di Matematica e Informatica, Università di Cagliari, Cagliari, Italy; \href{mailto:loi@unica.it}{loi@unica.it}}
\affil[2]{Dipartimento di Scienze Economiche e Aziendali, Università di Cagliari, Cagliari, Italy; \href{mailto:smatta@unica.it}{smatta@unica.it}}
\date{}

\begin{document}
\maketitle

\begin{abstract}
Suppose an economy moves continuously along a closed path of fundamentals and returns exactly
to its initial state. A regular Walrasian equilibrium price followed continuously along that
path need not return to itself. For every number of commodities $L\ge3$,
we construct a closed one-parameter family of pure-exchange economies in which all
preferences are fixed, only three individual endowments vary, and aggregate resources
remain constant. Along the loop, two regular equilibrium-price branches are exchanged while
a third returns to itself; at the end, every economic primitive is exactly as it was
initially. The construction uses at most $2L+2$ consumers. We also provide an exact
realization with at most $2L$ consumers in which all endowments are fixed and only one
preference varies; the two families generate identical aggregate excess demand. The key
realization problem is parametric: classical aggregate-realization results guarantee an
economy at each fixed parameter value, but not a continuous loop of realizing primitives.
We construct such a loop explicitly. The resulting monodromy is structurally stable under
sufficiently small $C^1$ perturbations. It is impossible with one or two commodities,
making three commodities the sharp threshold. Thus local regularity guarantees local
continuation but not a globally path-independent identity of an equilibrium branch, with
implications for global comparative statics under multiplicity and for equilibrium selection
based on continuation.
\end{abstract}



\section{Introduction}\label{sec:intro}

Suppose an economy moves continuously along a closed path of fundamentals and returns exactly
to its initial state. Must a regular Walrasian equilibrium price followed continuously along
that path also return to itself? At a regular equilibrium, the implicit function theorem gives
a unique local continuation for sufficiently small changes in fundamentals. It is therefore
natural to speak locally of following ``the same'' equilibrium as the economy changes. The
question addressed in this paper is whether these local identifications necessarily fit
together into a globally consistent identity of the equilibrium branch. We show that they do
not.

The phenomenon can be described without topological terminology. Start from an economy with
several regular equilibrium prices and select one of them. Change the fundamentals
continuously along a closed path and, at each small step, follow the uniquely continued
equilibrium. When the path is completed, every fundamental of the economy has returned
exactly to its initial value. Nevertheless, the equilibrium price reached by continuation need
not be the price from which we started: it can be another equilibrium price of the same
economy. The permutation of equilibrium prices produced by such a closed path is what we call
\emph{real monodromy}. Throughout the experiment both the economic primitives and the
equilibrium prices remain real.

Throughout the paper, $L$ denotes the number of commodities. Our main result shows that, for
every $L\ge3$, nontrivial real monodromy can occur in a finite pure-exchange economy while
\emph{all preferences remain fixed}. We construct a closed one-parameter family in which only
three individual endowments vary, their sum remains constant, and aggregate resources
therefore never change. Along the loop we track three regular equilibrium-price branches in a
compact region of strictly positive normalized prices. One returns to itself, while the other
two are exchanged. At the end of the experiment every preference and every endowment has
returned exactly to its initial value, yet an equilibrium price followed continuously around
the loop may have become another equilibrium price of that same economy. The construction
uses at most $2L+2$ consumers.

We also obtain a second exact realization, using at most $2L$ consumers, in which all
endowments are fixed and only one consumer's preference varies. The two families have
identical aggregate excess demand at every strictly positive price. The one-varying-preference
family is an intermediate constructive result; the fixed-preference redistribution family is
the stronger economic formulation. The monodromy is structurally stable under sufficiently
small $C^1$ perturbations of the reduced aggregate excess-demand system. The commodity
threshold is sharp: with one or two goods nontrivial real monodromy is impossible, whereas it
exists for every $L\ge3$.

The topological mechanism itself is classical: continuation around a loop can permute branches
of solutions, and such configurations can persist under small perturbations. The contribution
of the paper is its economic realization. We show that this geometry can occur in a finite
Walrasian pure-exchange economy and, more strongly, with fixed preferences, constant aggregate
resources, and parameter variation confined to the redistribution of endowments among three
consumers.

The result identifies a limitation of \emph{global comparative statics under multiplicity}.
Regularity is a local property: for a sufficiently small change in fundamentals, the implicit
function theorem identifies the unique nearby equilibrium that continues a given regular one.
Monodromy shows that concatenating such local continuations need not produce a globally
path-independent identity. Thus three logically distinct properties should not be conflated:
\[
\text{local regularity},\qquad
\text{global identifiability of a branch},\qquad
\text{uniqueness of equilibrium}.
\]
An economy may have several regular equilibria, each locally traceable, while two of the
resulting branches fail to admit globally consistent labels.

This distinction matters for equilibrium selection and for continuation-based counterfactuals.
If one selects an equilibrium by ``following it continuously from a benchmark,'' two parameter
paths connecting the same benchmark economy to the same target fundamentals can lead to
different equilibrium prices. The path, or an additional selection criterion, may therefore be
needed to identify the endpoint. Likewise, when a computational counterfactual is defined as
the continuation of a benchmark branch, reproducibility may require reporting the parameter
path as well as the initial and final economies. This issue is distinct from computing the
entire equilibrium set, where no single branch need be selected as the continuation of the
benchmark.

The transposition is not merely a relabelling of the same economic outcome. In our exact
construction the two exchanged prices imply different demands for a consumer whose own
preferences and endowment remain fixed along the loop. Hence the monodromy has allocative
content, although no welfare ranking between the two equilibria is implied.

Our construction belongs to the aggregate-demand realization program initiated by
Sonnenschein and developed by Mantel and Debreu. Sonnenschein established the broad freedom of
aggregate excess demand subject to the restrictions inherited from Walras' law
\cite{sonnenschein72,sonnenschein73}. Mantel supplied a constructive exact realization for a
large differentiable class, using at most two consumers per commodity \cite{mantel74}, while
Debreu obtained exact realization in the general continuous case with $L$ consumers
\cite{debreu74}.

For the present problem, the essential distinction is between \emph{pointwise} and
\emph{parametric} realization. At each fixed parameter value, Debreu's theorem already implies
that the formal excess-demand field used in our construction can be generated by an actual
pure-exchange economy on the relevant compact interior price region. But
\[
\text{for every parameter value there exists some realizing economy}
\]
does not imply
\[
\text{there exists one continuous family of realizing economies}.
\]
Monodromy requires the latter: the economic primitives themselves must form a continuous
closed loop.

We address this gap by adapting Mantel's constructive method to the parameterized setting and
controlling where the parameter enters individual primitives. For our special formal field, all
parameter dependence is concentrated in a single consumer's preference, while the aggregate
residual is independent of the parameter. Debreu's theorem is therefore used only once to
realize this fixed residual, producing a block of consumers that remains unchanged throughout
the loop. No pointwise selection of Debreu realizations is required as the parameter varies.

The final step is specific to our construction. The varying consumer's demand depends
affinely on two real parameter coefficients and is linear in wealth. This allows that consumer
to be replaced exactly by three consumers with fixed preferences, with all parameter variation
transferred to the shares in which a fixed total endowment is distributed among them. The
resulting fixed-preference family has exactly the same aggregate excess demand as the
one-varying-preference family.

Our analysis is formulated at the level of parameterized Walrasian excess-demand systems. The
monodromy mechanism therefore does not rely on a particular global parameterization of economies:
it requires only a regular family of equilibrium equations on the price region being followed.
This makes the construction portable across different economic realizations of the same local
equilibrium geometry. In the main result we realize it within pure exchange, first through
preference variation and then, more strongly, through endowment redistribution with fixed
preferences and constant aggregate resources. The prescribed-price construction shows further
that the same monodromy geometry can be built around arbitrary triples of positive equilibrium-
price rays, including configurations taken from economic models with different primitives.

The equilibrium-manifold literature of Balasko \cite{balasko88,balasko09} and Debreu's analysis
of regular economies \cite{debreu70,mascolell85} provide an important point of comparison. In
the fixed-preference endowment-redistribution setting, Loi and Matta \cite{loimatta10} use the
arc-lifting property of regular economies to show that, once an initial supporting equilibrium
price is fixed, a regular redistribution policy admits a unique continuous lift of equilibrium
prices. Our result concerns the complementary global issue: uniqueness of the lift along each
given path does not imply that the equilibrium reached by continuation is independent of the
path. Since our construction controls only a restricted covering over a compact price region,
we make no claim about monodromy of the full equilibrium covering.

The continuation interpretation also connects the paper to homotopy methods in computational
general equilibrium \cite{eaves99}, parameter homotopy methods \cite{hauenstein20}, and
monodromy-based methods in numerical algebraic geometry \cite{leykin09,duff19}. Hauenstein and
Regan develop a real monodromy action for real solution sets \cite{hr20real}; in our restricted
regular-equilibrium setting, the covering property yields an ordinary monodromy representation
into a permutation group. The two-good examples of Toda and Walsh are used to separate
equilibrium multiplicity from real monodromy \cite{todawalsh17}. The equilibrium fibres
computed by Kehoe and by Gauthier--Kehoe--Quintin are used only as price-and-index templates
for newly constructed pure-exchange families \cite{kehoe85num,gkq22}; no monodromy claim is
made about the source economies themselves.

The result is a possibility, robustness, and sharpness theorem, not a claim that monodromy is
generic or empirically prevalent. Nor does it rule out every global continuous selection: the
fixed branch in our construction is one. What fails is a global continuous identification of
either member of the mobile pair. The equilibrium structure is controlled on a compact
positive-price region; we do not claim that the three displayed branches exhaust the complete
equilibrium set or that equilibria outside that region are everywhere regular. The consumer
counts $2L$ and $2L+2$ are constructive upper bounds.

Finally, continuation here is a geometric operation on the equilibrium correspondence, not an
adjustment dynamics. It is not a t\^atonnement and does not determine which equilibrium an
actual decentralized market will select; in particular, monodromy alone does not imply dynamic
hysteresis or irreversibility. Its implication is structural: \emph{the identity supplied to a
regular equilibrium by local continuation need not extend to a globally path-independent
identity}.

The particular three prices used in the basic construction are also inessential. We show that
the same geometry can be realized with any prescribed triple of distinct positive price rays.
This permits the initial fibre to be chosen, for example, from equilibrium-price configurations
appearing in the general-equilibrium literature. The resulting monodromic family is a new
economy: it is not obtained by deforming the source economy and need not preserve its
preferences, endowments, allocations, production structure, stability properties, or other
equilibria.

Section~\ref{sec:prelim} introduces aggregate excess demand, normalized equilibrium prices,
regularity, restricted equilibrium-price coverings, and monodromy. Section~\ref{sec:formal}
constructs the formal Walrasian family with the required monodromy. Section~\ref{sec:exact}
explains why pointwise realization is not sufficient, and Sections~\ref{sec:realization}--
\ref{sec:endowment} build the continuous economic realizations: first with one varying
preference, then in arbitrary commodity dimension, and finally with fixed preferences and
varying endowments. Section~\ref{sec:main} states the main economic theorem.

Sections~\ref{sec:robust}--\ref{sec:economic} establish structural stability, prove the sharp
two-versus-three-good threshold, develop the prescribed-price construction, and discuss the
implications for comparative statics, equilibrium selection, and continuation algorithms.
Section~\ref{sec:conclusion} concludes. The appendices contain the detailed analytic proof
of the special realization, examples from the literature, and additional two-good examples
illustrating multiplicity without monodromy.

\section{Preliminaries}\label{sec:prelim}

This section introduces the objects needed to describe the global continuation of
Walrasian equilibrium prices. We first summarize the aggregate excess-demand setting and
normalize prices, so that an equilibrium price ray becomes an ordinary point in an
$(L-1)$-dimensional space. We then define regular equilibria and the finite coverings formed
by regular equilibrium branches in a parameterized family. Real monodromy is the
permutation obtained when those branches are followed around a closed loop of fundamentals.

Throughout the paper, $L$ denotes the number of commodities. The preliminary definitions
apply for $L\ge2$; the constructive monodromy results will require $L\ge3$. Prices lie in
$\R^L_{++}$ and consumption bundles in $\R^L_+$. We write $\one=(1,\ldots,1)$,
$\mathbf e_1,\ldots,\mathbf e_L$ for the coordinate unit vectors, and
$\widehat q=\diag(q_1,\ldots,q_L)$.

\begin{definition}[Exchange economy]\label{def:Meconomy}
For a fixed finite $L$, an \emph{exchange economy} consists of finitely many consumers,
each specified by the consumption set $X=\R^L_+$, a complete, transitive, continuous,
monotone and convex preference relation on $X$, and an endowment $\omega_i\in X$. The
aggregate endowment is strictly positive. At prices $p\gg0$, consumer $i$ has the compact
budget set $B(p,p\cdot\omega_i)=\{x\in X:p\cdot x\le p\cdot\omega_i\}$. Demand is the set
of maximal bundles; when a utility representation has a unique maximizer, demand is
identified with that bundle.
\end{definition}

As usual in pure exchange, the economy can be summarized at the aggregate level by its
excess-demand function. An \emph{aggregate excess-demand field} is a map
$Z\colon\R^L_{++}\to\R^L$ satisfying Walras' law $p\cdot Z(p)=0$ and homogeneity of degree
zero, $Z(\mu p)=Z(p)$ for $\mu>0$. We call such a field \emph{formal} when it is specified
directly, before showing that it is generated by utility-maximizing consumers with given
preferences and endowments.

Because only relative prices matter, one price can be used as numeraire. We normalize by
$p_L=1$ and set
\[
x_j=\frac{p_j}{p_L}-1\quad(j=1,\ldots,L-1),\qquad
u=x_1,\quad v=x_2,\quad \tau=(x_3,\ldots,x_{L-1}).
\]
The shift by one simply places the reference price vector $(1,\ldots,1)$ at the origin of
the normalized coordinates. A zero of the normalized reduced excess-demand system is a
\emph{normalized Walrasian equilibrium price}; it represents the corresponding positive
price ray. Unless a price-allocation equilibrium is explicitly mentioned, the terms
\emph{equilibrium}, \emph{equilibrium branch}, and \emph{equilibrium fibre} refer to these
normalized prices and their rays.

For a reduced system $F\colon\R^{L-1}\to\R^{L-1}$, a zero $x^*$ is \emph{regular} if
$D_xF(x^*)$ is invertible. Economically, regularity is what makes an equilibrium locally
traceable: the implicit function theorem gives a unique nearby equilibrium-price branch
for small changes in the parameters. The central question below is whether those locally
unique continuations always fit together into a globally consistent label. The reduced
equilibrium index is
\[
\ind(x^*)=\sgn\det\bigl(-D_xF(x^*)\bigr).
\]
We record the index because it will distinguish the fixed branch from the two mobile
branches in the construction. All regularity and index statements concern the reduced
system. The full derivative $D_pZ$ is singular because homogeneity implies
$D_pZ(p)p=0$.

We also record two elementary identities for a scalar hom-$0$ function
$f\colon\R^L_{++}\to\R$; they are used later in the realization argument. Differentiation
gives
\begin{equation}\label{eq:euler0}
q\cdot\nabla f(q)=0,
\end{equation}
and
\begin{equation}\label{eq:homgrad}
\nabla f(tq)=t^{-1}\nabla f(q),\qquad D^2f(tq)=t^{-2}D^2f(q).
\end{equation}

We now pass from one economy to a parameterized family. Think of $b\in B$ as describing the
fundamentals of the economy, and let $K$ be a fixed compact region of normalized positive
prices in which we track equilibria. Let $B$ be connected, locally path-connected and
locally compact Hausdorff, and let $F\colon U\times B\to\R^{L-1}$ be continuous, continuously
differentiable in $x\in U$, with $D_xF$ jointly continuous. Define
\[
E_K(F)=\{(x,b)\in K\times B:F(x,b)=0\},\qquad \pi_K(x,b)=b.
\]

\begin{definition}[Restricted regular equilibrium-price covering]\label{def:restrictedcover}
The family $F$ is a \emph{restricted regular equilibrium-price covering over $K$} if
\begin{enumerate}[label=\textup{(\roman*)},leftmargin=2.2em]
\item $F(x,b)\ne0$ for every $(x,b)\in\partial K\times B$;
\item every zero in $K\times B$ is regular; and
\item the projection $\pi_K\colon E_K(F)\to B$ is surjective.
\end{enumerate}
Under these conditions $\pi_K$ is a finite covering. We call its fibre the
\emph{restricted equilibrium-price fibre}.
\end{definition}

Thus ``restricted'' means that we track only equilibria inside $K$; ``regular'' means that
none of them becomes singular; and ``covering'' means that locally in the parameter they
move as a finite collection of separate continuous branches. Indeed, the parameterized
implicit function theorem makes $\pi_K$ a local homeomorphism. For every compact $C\subset
B$, the set $E_K(F)\cap(K\times C)$ is compact; hence $\pi_K$ is proper. A proper local
homeomorphism has finite fibres and is a covering. Since $B$ is connected, its degree is
constant.

Once equilibria form a covering, the global question has a precise formulation. Start from
one equilibrium price, move the economy around a closed path of fundamentals, and follow
that equilibrium continuously. When the economy returns to its initial fundamentals, the
continued price belongs again to the initial equilibrium fibre, but it need not be the
price from which we started.

\begin{definition}[Real monodromy of regular equilibrium prices]\label{def:monodromy}
Let $\pi\colon E\to B$ be a finite covering and let
$\gamma\colon[0,1]\to B$ be a loop based at $b_0$. For each
$e\in\pi^{-1}(b_0)$, path lifting gives a unique lift
$\widetilde\gamma_e$ with $\widetilde\gamma_e(0)=e$. Its endpoint belongs again to the
initial fibre. The permutation
\[
\rho([\gamma])(e)=\widetilde\gamma_e(1)
\]
is the \emph{real monodromy} along $\gamma$.
\end{definition}

Path lifting defines the monodromy representation
$\rho\colon\pi_1(B,b_0)\to\Sym(\pi^{-1}(b_0))$. A covering is connected if and only if
its monodromy acts transitively on a fibre \cite[Ch.~1.3]{hatcher02}.

\section{A formal model with real monodromy}\label{sec:formal}

This section isolates the geometry we later want to realize economically. We do not yet
construct preferences or endowments. Instead, we design a formal normalized price system
with exactly three regular equilibria: one remains fixed, while two are exchanged when the
parameter travels around a closed loop. We first build the mechanism in the two independent
relative-price coordinates available with three commodities, then embed it in every
$L\ge3$, and finally complete it to a formal Walrasian excess-demand field.

\subsection{A planar model with two exchanging equilibria}

With three commodities, normalization leaves two independent relative-price coordinates.
We combine them into the complex coordinate $z=u+iv$. Complex notation is used only as a
convenient representation of this two-dimensional real price space; no complex-valued
economic primitive is being introduced. The parameter $\lambda$ likewise represents two
real coordinates of the formal family. At this stage it is only a formal parameter; its
economic realization comes later.

Fix $c>0$ and identify the parameter plane with $\C$. For $\lambda\in\C$ define
\[
S(\lambda)=\lambda\bigl(|\lambda|^2-1\bigr)
=\sigma_1(\lambda)+i\sigma_2(\lambda).
\]
Let
\[
A=\{\lambda\in\C:1+\delta\le |\lambda|^2\le R^2\},
\qquad \delta>0,\ R>1.
\]
We use an annulus because it contains closed paths that wind around a hole. That winding
is what allows a locally defined square root to change sign after one complete circuit. On
$A$, $S$ never vanishes and $\pi_1(A)\cong\Z$.

We want exactly three zeros. Two should be the two square roots of a parameter-dependent
complex number, because those roots exchange when the parameter winds once around the
origin; the third should remain fixed. This motivates the planar generating map
\begin{equation}\label{eq:Phi}
\Phi(z;\lambda)=\bigl(z^2-S(\lambda)\bigr)(\bar z-c),\qquad z=u+iv.
\end{equation}
Here $\Phi\colon\R^2\to\R^2$ is real analytic and is written in complex notation; it is
not holomorphic. The conjugation has no economic interpretation: it is a mathematical
device chosen so that the fixed root has the opposite local index from the two mobile
roots. The zeros are the unordered mobile pair $\pm\sqrt{S(\lambda)}$ and the fixed root
$c$.

To speak unambiguously of three branches over the whole parameter region, we must ensure
that the roots never collide. We therefore choose the constants so that the mobile pair
remains separated from itself and from the fixed root uniformly over the annulus.

\begin{lemma}[Uniform separation]\label{lem:sep}
Suppose
$s_{\max}=R(R^2-1)=\sup_A|S|<c^2$ and put
$s_{\min}=\sqrt{1+\delta}\,\delta=\inf_A|S|>0$. Then the three roots are uniformly
distinct on $A$ and
\[
|z_+-z_-|\ge2\sqrt{s_{\min}},\qquad
\min_{\pm}|z_{\pm}-c|\ge c-\sqrt{s_{\max}}>0.
\]
\end{lemma}

\begin{proof}
$|S|=|\lambda|(|\lambda|^2-1)$ is increasing for $|\lambda|>1$. The stated bounds follow
from $|z_{\pm}|=\sqrt{|S|}$ and the reverse triangle inequality.
\end{proof}

\begin{remark}[Standing constants]\label{rem:numbers}
Choose $\delta,R,c,\rho$ and a transverse radius $\eta$ such that
\[
1+\delta<R^2,\qquad \sqrt{s_{\max}}<c<\rho<1,\qquad 0<\eta<1.
\]
For the numerical illustrations we use
\[
\delta=0.01,\qquad R^2=1.05,\qquad c=\tfrac12,\qquad
\rho=0.70.
\]
The auxiliary cutoff radii introduced later are chosen as
$\rho_1=0.72$ and $\rho_2=0.80$.
\end{remark}

The complex notation has now served its geometric purpose. To connect the construction with
normalized price equations and later with excess demand, we rewrite it as two real
equations in the two independent relative-price coordinates.

Write $z^2-S=\mathcal A+i\mathcal B$, where
$\mathcal A=u^2-v^2-\sigma_1$ and $\mathcal B=2uv-\sigma_2$. Expanding
\eqref{eq:Phi} gives $\Phi=F_1+iF_2$, with
\begin{align}
F_{1}(u,v;\lambda)&=u^{3}-cu^{2}+uv^{2}+cv^{2}-\sigma_{1}u+c\sigma_{1}-\sigma_{2}v,
\label{eq:F1}\\
F_{2}(u,v;\lambda)&=u^{2}v-2cuv+v^{3}+\sigma_{1}v-\sigma_{2}u+c\sigma_{2}.
\label{eq:F2}
\end{align}

The square-root mechanism already gives the monodromy. Along a generator of the annulus,
$S(\lambda)$ winds once around the origin, so a chosen square root acquires a minus sign.
The two mobile branches therefore exchange while the fixed branch remains fixed.

\begin{lemma}[Branch swap]\label{lem:swap}
Let $\gamma(t)=(r\cos t,r\sin t)$, $0\le t\le2\pi$, be a generator of $\pi_1(A)$. Then
continuation along $\gamma$ interchanges the two mobile zeros
$\pm\sqrt{S(\lambda)}$ of the planar system $(F_1,F_2)$ and fixes the third zero $c$.
\end{lemma}

\begin{proof}
Along $\gamma$, $S(\gamma(t))=r(r^2-1)e^{it}$. A square-root lift acquires the factor
$e^{i\pi}=-1$ after one circuit, while the root $c$ is constant.
\end{proof}

\subsection{Stabilization to arbitrary commodity dimension}

The monodromy mechanism uses only two independent relative-price coordinates, hence three
commodities. For larger commodity spaces we do not need a new topological construction. We
embed the same two-dimensional mechanism and force every additional normalized price
coordinate to remain at its reference value.

For $L\ge3$ define the reduced field
\begin{equation}\label{eq:FL}
F^{[L]}(u,v,\tau;\lambda)=
\bigl(F_1(u,v;\lambda),F_2(u,v;\lambda),-\tau\bigr)\in\R^{L-1},
\end{equation}
where the last block is absent when $L=3$. Its zeros are exactly
\[
(\sqrt{S(\lambda)},0),\quad(-\sqrt{S(\lambda)},0),\quad(c,0),
\]
with the complex number recording $(u,v)$. Since the transverse coordinates are
identically zero along all three branches, Lemma~\ref{lem:swap} immediately gives the same
branch transposition for $F^{[L]}$ in every dimension $L\ge3$.

We next verify the two local properties that will matter later: all three branches remain
regular in every commodity dimension, and the two branches that are exchanged have the
same local index while the fixed branch has the opposite index.

\begin{lemma}[Indices in every dimension]\label{lem:index}
For every $L\ge3$, the three zeros of $F^{[L]}$ are regular and have indices
$+1,+1,-1$, respectively.
\end{lemma}

\begin{proof}
At a mobile root $\zeta^2=S$, the planar differential is multiplication by
$2\zeta(\bar\zeta-c)\ne0$ and therefore has positive real determinant. At $z=c$ it is
multiplication by $c^2-S$ followed by conjugation and therefore has negative determinant.
Since the planar block has dimension two,
$\det(-D F_{\mathrm{pl}})=\det(D F_{\mathrm{pl}})$. The transverse derivative is
$-I_{L-3}$, hence
\[
-D F^{[L]}=
\begin{pmatrix}-D F_{\mathrm{pl}}&0\\0&I_{L-3}\end{pmatrix}.
\]
Thus the added coordinates preserve both regularity and the sign of
$\det(-D F)$.
\end{proof}

\begin{remark}[Index bookkeeping and conditional fibre minimality]\label{rem:conjugation}
The factor $\bar z-c$ reverses orientation at the fixed branch and gives it index $-1$;
the two mobile branches have index $+1$, so the restricted index sum is $+1$. The
transverse equation is $-\tau=0$, rather than $+\tau=0$, because its block in $-DF^{[L]}$
is $+I_{L-3}$ and introduces no parity-dependent sign. Had the last factor in
\eqref{eq:Phi} been $z-c$, all three indices would be positive and the restricted sum
would be $+3$, requiring equilibria of total index $-2$ outside the target region whenever
the global index theorem applies \cite{dierker72}. We do not invoke that theorem for the
restricted construction: the realization controls a compact target cone, not the entire
price simplex.

The index is locally constant on the regular equilibrium locus, so monodromy can permute
only sheets having the same index. This does not force three sheets for an arbitrary
restricted covering, because additional equilibria may lie outside the target region. If,
however, the restricted fibre is the complete finite regular equilibrium fibre, the global
index sum is $+1$ \cite{dierker72}. A nontrivial permutation then requires at least two
equilibria of the same index, while a two-point regular fibre has index sum in
$\{-2,0,2\}$, never $+1$. Under this additional global-completeness hypothesis, three is
the minimal fibre cardinality and the pattern $+1,+1,-1$ is minimal.
\end{remark}

\subsection{The Walrasian formal field and restricted covering}

So far we have only an abstract system of equations with the desired zero set. We now turn
it into an $L$-component field satisfying the two aggregate restrictions of Walrasian excess
demand: homogeneity of degree zero and Walras' law. This still does not mean that the field
is generated by actual consumers; microfoundation begins in Section~\ref{sec:exact}.

For $p\in\R^L_{++}$ let $x_j=p_j/p_L-1$ and define
\begin{equation}\label{eq:Z}
Z^{[L]}_j(p;\lambda)=F^{[L]}_j(x(p);\lambda)\quad(j=1,\ldots,L-1),\qquad
Z^{[L]}_L(p;\lambda)=-\sum_{j=1}^{L-1}\frac{p_j}{p_L}Z^{[L]}_j(p;\lambda).
\end{equation}

\begin{lemma}[Structural properties]\label{lem:struct}
$Z^{[L]}$ is $C^{\infty}$, homogeneous of degree zero, and satisfies Walras' law. On the
slice $p_L=1$, its equilibria are exactly the three zeros of $F^{[L]}$.
\end{lemma}

\begin{proof}
The normalized coordinates are hom-$0$, and the final component in \eqref{eq:Z} gives
$p\cdot Z^{[L]}=0$. If the first $L-1$ components vanish, Walras' law forces the last one
to vanish.
\end{proof}

We now isolate a compact normalized price region containing exactly the three branches.
The boundary condition below ensures that, as the parameter varies, no equilibrium can
enter or leave the controlled region through its boundary.

Set
\[
K_L=\{(u,v,\tau):u^2+v^2\le\rho^2,\ \|\tau\|_{\infty}\le\eta\},
\]
with the transverse condition omitted for $L=3$, and define
\[
\eco_{K_L}=\{(x,\lambda)\in K_L\times A:F^{[L]}(x;\lambda)=0\},\qquad
\pi_{K_L}\colon\eco_{K_L}\to A.
\]

\begin{lemma}[Boundary nonvanishing and covering]\label{lem:cover}
No zero of $F^{[L]}$ lies on $\partial K_L$. Consequently
$\pi_{K_L}$ is a three-sheeted covering for every $L\ge3$.
\end{lemma}

\begin{proof}
On the planar boundary $|z|=\rho$,
$|\Phi(z;\lambda)|\ge(\rho^2-s_{\max})(\rho-c)>0$. On a transverse boundary face some
coordinate of $-\tau$ is nonzero. The three zeros lie in the interior. Over a simply
connected subset of $A$, a branch of $\sqrt S$ gives three disjoint local sections.
\end{proof}

Geometrically, the restricted equilibrium set therefore consists of one globally
identifiable equilibrium and a pair that are locally distinguishable but globally
intertwined: going once around the annulus reverses the identities of the mobile pair.

\begin{proposition}[Structure of the equilibrium locus]\label{prop:structure}
For every $L\ge3$, $\eco_{K_L}$ is the disjoint union of a trivial one-sheeted covering
formed by the fixed equilibrium and a connected two-sheeted covering formed by the mobile
pair. Hence the mobile equilibria have no continuous global labelling over $A$.
\end{proposition}

\begin{proof}
The decomposition follows from the explicit zero set and uniform separation. The generator
acts on the mobile fibre by a transposition (Lemma~\ref{lem:swap}), so the double cover is
connected and admits no section.
\end{proof}

\begin{theorem}[Formal monodromy in every dimension]\label{thm:formal}
For every $L\ge3$, the formal family $Z^{[L]}$ has over $A$ a three-sheeted restricted
equilibrium covering with three regular zeros of indices $+1,+1,-1$. Its monodromy along a
generator is
\[
\rho_{K_L}([\gamma])=(q_+\ q_-),\qquad q_0\text{ fixed},
\]
and the image of the monodromy representation is isomorphic to $\Z/2\Z$.
\end{theorem}

\begin{proof}
Combine Lemmas~\ref{lem:index},~\ref{lem:swap}, and~\ref{lem:cover} with
Proposition~\ref{prop:structure}.
\end{proof}

\begin{remark}[A concrete initial fibre]\label{rem:prices}
With the numerical constants of Remark~\ref{rem:numbers} and
$\lambda_0=\sqrt{1.03}\in\R$, one has
$S(\lambda_0)=\sqrt{1.03}\,(0.03)\approx0.030447$. On the normalization $p_L=1$, the
three equilibrium prices are
\[
q_+\approx(1.17449,1,\ldots,1),\qquad
q_-\approx(0.82551,1,\ldots,1),\qquad
q_0=(1.5,1,\ldots,1).
\]
Continuation around the standard generator exchanges $q_+$ and $q_-$ and fixes $q_0$.
\end{remark}

\section{Why pointwise realization is not enough}\label{sec:exact}

The previous section produced a continuous family of formal aggregate excess-demand fields
with nontrivial monodromy. Classical aggregate-realization results already imply that, for
each fixed parameter value, the corresponding field can be generated by a pure-exchange
economy. This is not yet enough for our purpose. Choosing one realizing economy separately
at each parameter value does not imply that those economies can be chosen continuously as
the parameter moves.

Throughout the paper, an \emph{exact realization} of a formal aggregate excess-demand field
on a specified price domain means a pure-exchange economy whose aggregate excess-demand
function coincides with that field on the stated domain; no equality outside that domain is
implied unless explicitly stated. Thus realization refers to the construction of preferences
and endowments that microfound the prescribed aggregate field. A \emph{pointwise realization}
establishes this identity separately at each parameter value, whereas a \emph{parametric
realization} constructs a single parameterized family of primitives satisfying the identity
with the stated continuity in the parameter.

Pointwise realization answers the question ``does some economy generate this aggregate
field at this parameter value?'' Monodromy instead requires a single path of economic
primitives: the entire path of aggregate fields must be generated coherently by one
continuous family of economies. Thus there is no pointwise obstacle; every snapshot can
already be microfounded. The remaining difficulty is parametric.

The following theorem isolates exactly what the classical realization result gives us---and
what it does not.

\begin{theorem}[Pointwise exact aggregate realization]\label{thm:aggregate}
Fix $L\ge3$ and $\lambda\in A$. There exists an $L$-consumer pure-exchange economy whose
aggregate excess demand equals $Z^{[L]}(\,\cdot\,;\lambda)$ on the positive-price cone generated by $K_L$.
Consequently, at that parameter, its restricted equilibrium fibre, reduced Jacobians, and
local indices coincide with those of the formal field. This pointwise statement does not
select the realizing economies continuously in $\lambda$ and therefore does not by itself
produce an economic monodromy theorem.
\end{theorem}

\begin{proof}
The field $Z^{[L]}(\,\cdot\,;\lambda)$ is continuous, hom-$0$, and Walrasian. Its normalized
price region is compactly contained in the interior of the simplex. Debreu's exact
realization theorem \cite{debreu74}, in the compact-interior form stated as
Theorem~\ref{thm:debreu}, gives $L$ consumers whose excess demands sum to this field on a
truncated simplex containing the image of the target cone. Homogeneity extends the identity
to the cone. Exact equality transfers all restricted zero-set data. The theorem is pointwise: no continuity of the selected consumers in $\lambda$ is asserted.
\end{proof}

What is missing is therefore not realization, but continuous parametric realization.

\section{Parametric realization in dimension three}\label{sec:realization}

The previous section identified the missing step: we need to realize the entire formal
family by one continuous family of economies. We now do this explicitly in the minimal
three-commodity case. The construction keeps all endowments fixed and all but one
preference fixed. Only one consumer changes with the parameter, and on the target price
region the resulting aggregate excess demand agrees exactly with the formal field.

The proof has two conceptual stages. First, we decompose the aggregate field into scalar
potentials in such a way that all parameter dependence is concentrated in one potential.
Second, we adapt Mantel's constructive realization to turn those potentials into individual
demands. The explicit block leaves a parameter-independent aggregate residual; Debreu's
theorem is then applied once to cancel that fixed residual. Thus no pointwise selection of
Debreu economies is made along the parameter path.

In this section only, write $Z=Z^{[3]}$, $K=K_3$, and normalize by $p_3=1$.

\subsection{Technical domains and standing constants}

The construction needs some room around the region on which exact realization is required.
We therefore place the controlled equilibrium region inside a larger target region, and
that target region inside a still larger compact region where auxiliary potentials can be
modified without approaching the boundary of the positive price simplex. The same idea is
used in parameter space. The cutoff below leaves the construction unchanged on the target
region and modifies it only outside.

\begin{assumption}[Standing data]\label{ass:data}
Fix once and for all:
\begin{enumerate}[label=\textup{(D\arabic*)},leftmargin=2.6em]
\item radii $0<\rho<\rho_{1}<\rho_{2}<1$;
\item the closed disks $\overline D_{r}=\{(u,v):u^{2}+v^{2}\le r^{2}\}\subset\R^{2}$ and the
cones
\[
\Gamma_{r}\deq\bigl\{p\in\R^{3}_{++}:\ (u(p),v(p))\in\overline D_{r}\bigr\},\qquad r\in\{\rho,\rho_{1},\rho_{2}\};
\]
the \emph{target cone} is $\Gamma\deq\Gamma_{\rho_{1}}$, and the price region $K$ of Section~\ref{sec:formal} is recovered as $\Gamma_{\rho}\cap\{p_{3}=1\}$, so that the cone over $K$ is contained
in the interior of $\Gamma$;
\item two \emph{compact} sets $\Lambda,\Lambda_{1}\subset\R^{2}$ and the \emph{open} set
$W\deq\operatorname{int}\Lambda_{1}$, with $\Lambda\subset W$ (in the application,
$\Lambda$ is a compact neighborhood of the annulus $A$ and $\Lambda_{1}$ a slightly larger
one). All constants below are taken uniform over $\Lambda_{1}$; all
smoothness-in-$\lambda$ statements hold on the open set $W\supset\Lambda$. (If $\Lambda$
itself has empty interior this costs nothing: the data are polynomial in $\lambda$, so any
enlargement is available.)
\item a cutoff $\chi\in C^{\infty}(\R^{2},[0,1])$ with
\[
\chi\equiv1\ \text{on an open set containing }\overline D_{\rho_{1}},\qquad
\supp\chi\subset \overline D_{\rho_{2}}.
\]

\end{enumerate}
Since $\rho_{2}<1$, the normalized prices on $\supp\chi$ remain in a compact subset of
$\R^{3}_{++}$.
\end{assumption}

\subsection{Continuity of the economic family}

We work with the exchange-economy class of Definition~\ref{def:Meconomy}. To turn the
pointwise statement of Section~\ref{sec:exact} into a genuine parameterized family, we must
also specify what it means for the single varying preference to depend continuously on the
parameter. The topology below is used only for that represented preference family; the
fixed-preference endowment realization constructed later lives in an ordinary
finite-dimensional Euclidean endowment space.

\begin{definition}[Topology for represented preference families]\label{def:reptopology}
A \emph{represented economy} records the chosen utility representatives for the explicitly
constructed block, the preference relations of the residual block, and all endowments. We
equip $C(\R^L_+)$ with the compact-open topology, the endowment space with its Euclidean
topology, and hold the residual preferences fixed. A parameterized represented family is
continuous when the chosen utility representatives and endowments vary continuously in
this product topology. For the purposes of the construction, no separate
quotient topology on ordinal preference relations is required.
\end{definition}

\begin{remark}[Role of utility representatives]\label{rem:representatives}
The equilibrium correspondence and its monodromy are ordinal objects: they depend
on the consumers' preference relations, not on a particular cardinal normalization
of utility. The chosen utility representatives are used only to formulate a concrete
and sufficiently strong topology for the explicitly constructed family. Along the
closed loop, the varying representative returns exactly to its initial value, and
therefore so does the associated preference relation. All other preference relations
and all endowments remain fixed throughout the loop. Thus the use of represented
economies does not alter the economic content of the monodromy statement.
\end{remark}

\begin{remark}[The two-level regularity convention]\label{rem:whyclass}
Throughout this document a single convention is maintained \emph{without exception}: global
properties on the closed orthant are the \emph{weak} ones (continuity, monotonicity,
quasiconcavity), while all \emph{strict} and \emph{smooth} properties (joint $C^{\infty}$
regularity, strong monotonicity, strict quasiconcavity) are asserted only on the open
orthant $\R^{3}_{++}$, where---as will be proved---the demand always lies at strictly
positive prices and positive wealth. Concretely, the varying-block utilities are built as
smooth functions $u^{i}_{\lambda}$ on the open orthant (Section~\ref{sec:direct}) and then
extended to $X$ as $U^{i}_{\lambda}=e^{u^{i}_{\lambda}}$, prolonged by $0$ on the boundary
(Proposition~\ref{prop:extension}); the extension changes nothing in the interior analysis,
because on $\R^{3}_{++}$ maximizing $U^{i}_{\lambda}$ and maximizing $u^{i}_{\lambda}$ are
the same problem.
\end{remark}

\subsection{Main three-good realization theorem}

The result can be summarized before its technical formulation: the three-good formal family
is realized exactly by six consumers; all endowments are fixed, five consumers have fixed
preferences, and only consumer $3$ varies continuously with the parameter.

\begin{theorem}[Explicit three-good parametric Mantel realization]\label{thm:parametric}
Under Assumption~\ref{ass:data} there exist a constant $k>0$, a cutoff $\chi$,
utilities $U^{1},U^{2}$, preference relations $\succeq^{4},\succeq^{5},\succeq^{6}$, and
endowments $\omega^{1},\dots,\omega^{6}$---all chosen once and for all, independently
of $\lambda$---together with a family of utilities $U^{3}_{\lambda}$, such that for each
$\lambda\in\Lambda$ the list
\[
\eco_{\lambda}=\bigl(U^{1},\,U^{2},\,U^{3}_{\lambda},\,
\succeq^{4},\,\succeq^{5},\,\succeq^{6};\
\omega^{1},\dots,\omega^{6}\bigr)
\]
is a six-consumer exchange economy in the standard class of
Definition~\ref{def:Meconomy}, with the following properties.
\begin{enumerate}[label=\textup{(\alph*)},leftmargin=2.2em]
\item \textbf{Varying block (homotheticity and two-level regularity).} For
$i=1,2,3$, the utility $U^{i}_{\lambda}\colon\R^{3}_{+}\to\R$ is nonnegative,
continuous, monotone, concave, and positively homogeneous of degree one on the closed
orthant, with $U^{i}_{\lambda}=0$ exactly on the boundary; for $i=3$ the map
$(x,\lambda)\mapsto U^{3}_{\lambda}(x)$ is jointly continuous on
$\R^{3}_{+}\times\Lambda$. The restriction of $U^{i}_{\lambda}$ to the open orthant is
$C^{\infty}$ and strongly monotone and strictly quasiconcave. Its logarithm
$u^{i}_{\lambda}=\log U^{i}_{\lambda}$ is $C^{\infty}$, strongly monotone, and strictly
concave; for $i=3$, $(x,\lambda)\mapsto u^{3}_{\lambda}(x)$ is jointly $C^{\infty}$ on
$\R^{3}_{++}\times W$. No strict property is asserted on boundary faces; those faces are never reached by the
varying block's demand at strictly positive prices and positive wealth. The endowments are
$\omega^{i}=k\,\mathbf e_i\in X$, so each endowment belongs to its consumer's consumption set,
and already
\[
\sum_{i=1}^{3}\omega^{i}=k\,\one\gg0 ,
\]
so the aggregate-resources condition of Definition~\ref{def:Meconomy} holds independently
of the residual block's endowments.
\item \textbf{Residual block.} Consumers $4,5,6$ are specified by $\lambda$-independent
preference relations $\succeq^{4},\succeq^{5},\succeq^{6}$ on $X=\R^{3}_{+}$ that are
continuous, monotone (in the sense of Debreu's construction \cite{debreu74}) and strictly convex, with
$\lambda$-independent endowments; they are supplied by Debreu's theorem
(Theorem~\ref{thm:debreu} below) applied once, to a single $\lambda$-free field, and they
belong to the class of Definition~\ref{def:Meconomy} natively. (Continuous utility
representations exist by Debreu's representation theorem, but are not needed.)
\item \textbf{Exact realization on the target region.} For every $\lambda\in\Lambda$ and
every $p\in\Gamma$,
\[
\sum_{i=1}^{6} z^{i}_{\lambda}(p)\;=\;Z(p;\lambda),
\]
where $z^{i}_{\lambda}(p)=x^{i}_{\lambda}(p)-\omega^{i}$ and $x^{i}_{\lambda}(p)$ is
consumer $i$'s demand at prices $p$ and wealth $p\cdot\omega^{i}$. In particular the
identity holds on the cone over $K$. The identity is asserted \emph{only} on $\Gamma$: off
the open set where the cutoff $\chi$ is identically $1$, the auxiliary gauge components are
modified and the aggregate need not equal $Z$.
\item \textbf{Parametric regularity; only one consumer varies.} The individual demands of
the varying block, $(p,m,\lambda)\mapsto x^{i}_{\lambda}(p,m)$ ($i=1,2,3$), are
$C^{\infty}$ jointly for $p\gg0$, $m>0$; all mixed derivatives
$\partial^{\beta}_{\lambda}\partial^{\alpha}_{x}u^{3}_{\lambda}(x)$ exist and are jointly
continuous, uniformly on compact subsets of $\R^{3}_{++}\times\Lambda$ (equivalently, by
the exponential law for spaces of smooth maps \cite{km97}, $\lambda\mapsto u^{3}_{\lambda}$ is a
smooth map from $W$ into $C^{\infty}(\R^{3}_{++})$ with the compact-open $C^{\infty}$
topology), and $(x,\lambda)\mapsto U^{3}_{\lambda}(x)$ is jointly continuous on the closed
orthant. All the remaining data---$U^{1}$, $U^{2}$, the relations
$\succeq^{4},\succeq^{5},\succeq^{6}$, the endowments $\omega^{1},\dots,\omega^{6}$, the
constant $k$, and the cutoff $\chi$---are fixed once and do not depend on $\lambda$: the
parameter enters the economy through consumer $3$'s preferences alone.
\end{enumerate}
\end{theorem}

Throughout, ``demand'' means the unique maximizer of the utility on the compact budget set
$B(p,m)=\{x\in\R^{3}_{+}:p\cdot x\le m\}$; existence, interiority and uniqueness for the
varying block are part of what is proved (Theorem~\ref{thm:rationalize} and
Lemma~\ref{lem:localization}).

\subsection{Concentrating the parameter in one consumer}\label{sec:potential}

Mantel's construction turns suitable scalar functions of prices into individual demands.
To realize an aggregate field with several consumers, we therefore seek scalar potentials
whose wealth-weighted gradients add up to the prescribed excess demand. Walras' law makes
such a decomposition easy to obtain. What is not automatic is the parameter structure: if
several potentials depend on the parameter, several consumers will vary. Our goal is to
choose the decomposition so that exactly one potential carries all parameter dependence.
Each potential will later become one consumer, so this concentration is precisely what
produces a one-varying-preference economy.

\begin{lemma}[Walras' law as a canonical weighted-gradient identity]\label{lem:walraspotential}
Let $Z\colon\R^L_{++}\to\R^L$ be $C^1$, homogeneous of degree zero, and Walrasian. Then
\[
\sum_{i=1}^L p_i\nabla Z_i(p)=-Z(p).
\]
Consequently the canonical choice $\varphi^i=Z_i$ solves the Mantel-type weighted-gradient
equation $\sum_i p_i\nabla\varphi^i=-Z$ that underlies Mantel's construction.
\end{lemma}

\begin{proof}
Differentiate $\sum_i p_iZ_i(p)=0$ with respect to $p_j$. This gives
$Z_j(p)+\sum_i p_i\partial_jZ_i(p)=0$ for every $j$, which is the stated vector identity.
\end{proof}

The lemma explains both what is easy and what is new. A potential solution always exists in
closed form, but the canonical choice inherits the parameter dependence of every component
of $Z$. In the present field this would make several consumers vary. The explicit
decomposition below is not needed for existence; its role is to place the parameter in one
potential only. This is the parsimonious step that pointwise SMD realization does not
provide and that later makes the endowment conversion possible.

The required special solution is written first in the equivalent sign convention
\begin{equation}\label{eq:poteq}
\sum_{i=1}^{3}p_{i}\,\nabla_{p}\,\widehat g_{i}(p;\lambda)\;=\;Z(p;\lambda)
\qquad(p\in\R^{3}_{++}).
\end{equation}
Equation \eqref{eq:poteq} converts the wealth-weighted gradients of Mantel's construction
into the target field. Only the final potential below depends on the parameter.

\begin{proposition}[Explicit potentials]\label{prop:potential}
Define, as functions of $(u,v)$ with parameters $(\sigma_{1},\sigma_{2})$,
\begin{align}
g_{1}(u,v)&=-2cv^{2},\qquad g_{2}=0,\label{eq:g1}\\
g_{3}(u,v;\lambda)&=\frac{u^{4}}{4}-\frac{c}{3}u^{3}+\frac{u^{2}v^{2}}{2}+cuv^{2}
-\frac{\sigma_{1}}{2}u^{2}+c\sigma_{1}u-\sigma_{2}uv
+\frac{v^{4}}{4}+\frac{\sigma_{1}}{2}v^{2}+c\sigma_{2}v+2cv^{2},\label{eq:g3}
\end{align}
and set $\widehat g_{i}(p;\lambda)\deq g_{i}\bigl(u(p),v(p);\lambda\bigr)$. Then each
$\widehat g_{i}$ is $C^{\infty}$ on $\R^{3}_{++}\times\R^{2}$, homogeneous of degree zero in
$p$, polynomial in $(u,v,\sigma_{1},\sigma_{2})$, with $g_{1},g_{2}$ independent of
$\lambda$, and identity \eqref{eq:poteq} holds for all $p\in\R^{3}_{++}$ and all
$\lambda\in\R^{2}$.
\end{proposition}

The explicit polynomial in \eqref{eq:g3} is less important than its parameter structure:
$g_1$ and $g_2$ are fixed, while all dependence on $\lambda$ has been confined to $g_3$.
That is the field-specific algebraic step on which the one-consumer concentration rests.

\begin{proof}
Homogeneity and smoothness are clear since $u(p),v(p)$ are hom-$0$ and smooth and
$g_{i}$ are polynomials. We verify \eqref{eq:poteq}. Write $q_{1}=p_{1}/p_{3}=1+u$,
$q_{2}=p_{2}/p_{3}=1+v$. For a hom-$0$ function $\widehat g(p)=g(u,v)$, the chain rule gives
\[
\frac{\partial \widehat g}{\partial p_{1}}=\frac{g_{u}}{p_{3}},\qquad
\frac{\partial \widehat g}{\partial p_{2}}=\frac{g_{v}}{p_{3}},\qquad
\frac{\partial \widehat g}{\partial p_{3}}=-\frac{p_{1}}{p_{3}^{2}}g_{u}-\frac{p_{2}}{p_{3}^{2}}g_{v}.
\]
Hence the left side of \eqref{eq:poteq} has components
\begin{align*}
\Bigl(\sum_{i}p_{i}\nabla\widehat g_{i}\Bigr)_{1}
&=(1+u)\,g_{1,u}+(1+v)\,g_{2,u}+g_{3,u},\\
\Bigl(\sum_{i}p_{i}\nabla\widehat g_{i}\Bigr)_{2}
&=(1+u)\,g_{1,v}+(1+v)\,g_{2,v}+g_{3,v},\\
\Bigl(\sum_{i}p_{i}\nabla\widehat g_{i}\Bigr)_{3}
&=-(1+u)\Bigl(\sum_{i}p_{i}\nabla\widehat g_{i}\Bigr)_{1}
-(1+v)\Bigl(\sum_{i}p_{i}\nabla\widehat g_{i}\Bigr)_{2},
\end{align*}
where the third line follows by collecting the $\partial/\partial p_{3}$ terms exactly as in
the first two. Comparing with \eqref{eq:Z}, identity \eqref{eq:poteq} is therefore
equivalent to the two scalar identities
\begin{equation}\label{eq:2eq}
(1+u)g_{1,u}+(1+v)g_{2,u}+g_{3,u}=F_{1},\qquad
(1+u)g_{1,v}+(1+v)g_{2,v}+g_{3,v}=F_{2}.
\end{equation}
With the ansatz $g_{2}=0$ and $g_{1}=g_{1}(v)$ (so $g_{1,u}=0$), \eqref{eq:2eq} reads
$g_{3,u}=F_{1}$ and $g_{3,v}=F_{2}-(1+u)g_{1}'(v)$; a $C^{2}$ solution $g_{3}$ exists iff
the right sides form a closed $1$-form, i.e.
\[
\partial_{v}F_{1}=\partial_{u}\bigl[F_{2}-(1+u)g_{1}'(v)\bigr]
=\partial_{u}F_{2}-g_{1}'(v).
\]
From \eqref{eq:F1}--\eqref{eq:F2} one computes $\partial_{u}F_{2}-\partial_{v}F_{1}=-4cv$,
so $g_{1}'(v)=-4cv$ and $g_{1}=-2cv^{2}$, which is \eqref{eq:g1}. Integrating along the
$L$-shaped path $(0,0)\to(u,0)\to(u,v)$ (the disk is convex, hence the potential is
well-defined) gives \eqref{eq:g3}. Finally one checks by direct differentiation of
\eqref{eq:g3} that
\[
g_{3,u}=F_{1},\qquad g_{3,v}=F_{2}+4cv(1+u),
\]
which is \eqref{eq:2eq}.
\end{proof}

\subsection{From potentials to consumers}

The explicit potentials are the only part of the construction that is specific to the formal field. Once the parameter has been concentrated in $g_3$, the remaining steps implement a uniform parametric version of Mantel's gauge method.
A homogeneous cutoff first localizes the potentials inside the positive price simplex while
preserving the exact potential identity on the target cone. A sufficiently large common
constant $k$ then makes the logarithmic gauges uniformly strictly convex and coordinatewise
monotone for all parameters. Their negative gradients are global diffeomorphisms, whose
inverses define smooth, strongly monotone and strictly quasiconcave direct utilities. The
resulting first block has the prescribed excess demand plus a fixed universal residual
$kR$; one parameter-free application of Debreu's theorem cancels that residual.

Uniform estimates are taken over a compact neighborhood of the parameter annulus, so the
utilities and demands vary jointly smoothly in the interior and continuously on the closed
orthant. Since only $\varphi^3_\lambda$ depends on the parameter, the completed family has
fixed endowments and only consumer $3$ varies. Appendix~\ref{app:analytic} gives the full
cutoff, global-inverse, boundary-extension, localization, residual-block, and assembly
arguments. It proves Theorem~\ref{thm:parametric} without any appeal to a pointwise choice
of realizing economies.

\section{Parametric realization in every commodity dimension}\label{sec:stabilization}

The three-good case contains the entire nontrivial geometry. Additional commodities do not
create new monodromy; they add transverse normalized price coordinates that must simply be
kept at their reference values. We now show that the potential decomposition and the
Mantel realization extend to every $L\ge3$ without introducing any new parameter-dependent
consumer.

In three goods the prescribed potential already reproduces the two active price
coordinates $(u,v)$. For every additional commodity we only need to reproduce the
transverse equation $-x_j=0$. This is achieved by adding the elementary quadratic term
$-x_j^2/2$ to the same parameter-dependent potential, since its derivative is $-x_j$.

Let $g_{\mathrm{pl}}(u,v;\lambda)$ denote the polynomial called $g_3$ in
\eqref{eq:g3}. For $L\ge3$, with $\tau=(x_3,\ldots,x_{L-1})$, define
\begin{equation}\label{eq:gL}
g^{[L]}_1=-2cv^2,\qquad g^{[L]}_i=0\ (2\le i\le L-1),\qquad
g^{[L]}_L=g_{\mathrm{pl}}(u,v;\lambda)-\frac12\sum_{j=3}^{L-1}x_j^2,
\end{equation}
where the sum is empty for $L=3$. Let
$\widehat g^{[L]}_i(p;\lambda)=g^{[L]}_i(x(p);\lambda)$.

\begin{proposition}[Explicit potentials in arbitrary dimension]\label{prop:potentialL}
For every $L\ge3$,
\begin{equation}\label{eq:poteqL}
\sum_{i=1}^{L}p_i\nabla_p\widehat g^{[L]}_i(p;\lambda)=Z^{[L]}(p;\lambda)
\qquad(p\in\R^L_{++}).
\end{equation}
Only $\widehat g^{[L]}_L$ depends on $\lambda$.
\end{proposition}

Thus the higher-dimensional extension does not spread the parameter across consumers: only
the final potential continues to depend on $\lambda$.

\begin{proof}
For a hom-$0$ function $\widehat g(p)=g(x_1,\ldots,x_{L-1})$,
\[
\partial_{p_j}\widehat g=p_L^{-1}g_{x_j}\quad(j<L),\qquad
\partial_{p_L}\widehat g=-p_L^{-1}\sum_{j<L}(1+x_j)g_{x_j}.
\]
Therefore the $j$th component, $j<L$, of the left side of \eqref{eq:poteqL} is
\[
\sum_{i=1}^{L-1}(1+x_i)\,\partial_{x_j}g^{[L]}_i
+\partial_{x_j}g^{[L]}_L.
\]
For $j=1,2$ these are exactly the two identities proved in
Proposition~\ref{prop:potential}. For $j\ge3$ the expression is
$\partial_{x_j}g^{[L]}_L=-x_j$, the transverse component of $F^{[L]}$. The last component
is the Walrasian completion of the first $L-1$ components by Euler's identity, hence equals
$Z^{[L]}_L$. Parameter dependence occurs only through $g_{\mathrm{pl}}$ in the last
potential.
\end{proof}

Set $\varphi^{[L]}_i=-g^{[L]}_i$. Then
\begin{equation}\label{eq:minusZL}
\sum_{i=1}^{L}p_i\nabla_p\widehat\varphi^{[L]}_{i,\lambda}(p)=-Z^{[L]}(p;\lambda),
\end{equation}
and only $\varphi^{[L]}_L$ varies with $\lambda$.

The economic takeaway is unchanged in every dimension: for $L\ge3$ the formal field can
be realized with at most $2L$ consumers, all endowments fixed, and only one varying
preference.

\begin{theorem}[Explicit finite-agent realization with one varying consumer]\label{thm:parametricL}
Let $L\ge3$. There is a closed target cone $\Gamma_L$ whose normalized image is compactly
contained in the positive simplex and contains the cone over $K_L$, together with the
following continuous family of $N=2L$ exchange economies:
\[
\eco^{[L]}_\lambda=
\bigl(U^1,\ldots,U^{L-1},U^L_\lambda,
\succeq^{L+1},\ldots,\succeq^{2L};
\omega^1,\ldots,\omega^{2L}\bigr),\qquad \lambda\in A,
\]
such that:
\begin{enumerate}[label=\textup{(\alph*)},leftmargin=2.2em]
\item all endowments and all primitives except $U^L_\lambda$ are independent of $\lambda$;
\item the first $L$ utilities are nonnegative, continuous, monotone, concave, and
positively homogeneous of degree one on $\R^L_+$; their logarithms are strongly monotone,
strictly concave, and smooth on $\R^L_{++}$, with the same joint parameter regularity as in
Theorem~\ref{thm:parametric};
\item the residual $L$ consumers have fixed continuous, monotone, strictly convex
preferences;
\item the aggregate excess demand equals $Z^{[L]}(\,\cdot\,;\lambda)$ exactly on
$\Gamma_L$.
\end{enumerate}
The family is continuous in the topology of Definition~\ref{def:reptopology}. Thus
$N\le2L$ is a constructive, not necessarily minimal, bound.
\end{theorem}

\begin{proof}
Choose a smooth hom-$0$ cutoff in the normalized coordinates that equals one on a
neighborhood of $\Gamma_L$ and has support compactly contained in the positive normalized
simplex. Apply it to the cores $\varphi^{[L]}_i$, and denote the resulting hom-$0$ cutoff
extensions by $\ph^{\,i}_\lambda$. Their scaled first and second derivatives
are uniformly bounded over the compact parameter set. For $k>0$ define the $L$-dimensional
Mantel gauges
\[
\psi^i_\lambda(q)=\frac{\ph^{\,i}_\lambda(q)}{k}
-\frac1L\sum_{j=1}^{L}\log q_j.
\]
Exactly as in Lemmas~\ref{lem:convexity} and~\ref{lem:positivity}, there are finite uniform
constants $C,C_1$ such that
\[
D^2\psi^i_\lambda(q)\succeq
\left(\frac1L-\frac{C}{k}\right)\widehat q^{-2},\qquad
-q_j\partial_j\psi^i_\lambda(q)\ge\frac1L-\frac{C_1}{k}.
\]
One choice $k>L\max\{C,C_1\}$ works for all consumers and parameters. The radial identity
$q\cdot\nabla\psi^i_\lambda=-1$ is unchanged. The proof of
Proposition~\ref{prop:diffeo} therefore gives a global diffeomorphism
$-\nabla\psi^i_\lambda\colon\R^L_{++}\to\R^L_{++}$, and the direct-utility construction of
Section~\ref{sec:direct} gives a unique demand
\begin{equation}\label{eq:demandL}
x^i_\lambda(p,m)=\frac{m}{L}\widehat p^{-1}\one
-\frac{m}{k}\nabla\ph^{\,i}_\lambda(p).
\end{equation}
The boundary argument is dimension-free: the Cobb--Douglas sandwich in
Proposition~\ref{prop:boundary} becomes a comparison with
$(x_1\cdots x_L)^{1/L}$, so the utilities extend continuously by zero to
$\partial\R^L_+$. The log-homogeneity argument of
Corollary~\ref{cor:homothetic} gives degree-one homogeneity of the extended utilities, and
quasiconcavity plus nonnegativity then gives concavity. Demand remains interior at positive
prices and wealth.

Give the first $L$ consumers the fixed endowments $\omega^i=k\mathbf e_i$. Summing
\eqref{eq:demandL} at wealth $kp_i$ yields
\[
\sum_{i=1}^{L}z^i_\lambda(p)=kR_L(p)
-\sum_{i=1}^{L}p_i\nabla\ph^{\,i}_\lambda(p),\qquad
R_L(p)=\frac{\one\cdot p}{L}\widehat p^{-1}\one-\one.
\]
On $\Gamma_L$, the cutoff is flat and \eqref{eq:minusZL} gives
$\sum_{i=1}^{L}z^i_\lambda=Z^{[L]}+kR_L$. The field $R_L$ is hom-$0$, Walrasian, and
parameter independent. Apply Theorem~\ref{thm:debreu} once to $-kR_L$ on the normalized
image of $\Gamma_L$; this supplies $L$ fixed residual consumers. Their aggregate cancels
$kR_L$, proving exact realization. Since only the last core depends on $\lambda$, only
$U^L_\lambda$ varies. Joint interior smoothness and compact-open continuity on the closed
orthant follow from the same inverse-map and boundary estimates used in
Theorem~\ref{thm:parametric}.
\end{proof}

\begin{remark}[A familiar fixed core of the construction]\label{rem:cdfixed}
For $2\le i\le L-1$ one has $g_i^{[L]}=0$. Hence these consumers are exactly symmetric
Cobb--Douglas consumers with utility proportional to
$\frac1L\sum_{j=1}^L\log x_j$, fixed preferences, and fixed vertex endowments
$k\mathbf e_i$. Only consumers $1$ and $L$ in the explicit block use nonzero gauge
perturbations, and only consumer $L$ carries the preference parameter.
\end{remark}

\begin{remark}[Homotheticity without representative-consumer aggregation]
The explicitly constructed preferences are homothetic but deliberately heterogeneous.
Identical homothetic preferences with unique demands would aggregate through wealth
linearity to a representative demand. Here the distinct gauge potentials assemble the
prescribed aggregate field; homotheticity supplies wealth scaling but does not collapse the
economy to a representative consumer. This is consistent with the aggregate flexibility of
heterogeneous homothetic preferences established by Mantel \cite{mantel76}.
\end{remark}

\section{From preference variation to endowment redistribution}\label{sec:endowment}

The previous sections realize monodromy with a single varying preference. We now prove the
economically stronger result that preference variation is not needed at all. The varying
consumer will be replaced exactly by three consumers whose preferences never change. All
parameter dependence will be transferred to the shares in which the original endowment is
distributed among them, while their total endowment remains constant.

The mechanism is simple once its two linearities are visible. The varying consumer depends
on the economic parameter only through two real coefficients
$\sigma=(\Re S(\lambda),\Im S(\lambda))$. If his demand is linear in wealth and affine in
$\sigma$, then any coefficient $\sigma$ lying inside a triangle can be written as a
barycentric average of three fixed coefficient values. Splitting the consumer's endowment
in the same barycentric shares reproduces exactly the demand of the original varying
consumer. The technical statements below establish precisely these two linearities.

The parameter $\lambda$ does not enter the varying consumer in an arbitrary way. It enters
through the two real coefficients of the complex number $S(\lambda)$. We therefore treat
$\sigma=(\sigma_1,\sigma_2)$ temporarily as an independent two-dimensional coefficient; no
new economic parameter is being introduced.

For $\sigma=(\sigma_1,\sigma_2)\in\R^2$, write
$Z^{[L]}(p;\sigma)$ for the field obtained from \eqref{eq:FL}--\eqref{eq:Z} by replacing
$S(\lambda)$ with the independent complex coefficient $\sigma_1+i\sigma_2$. The
construction above does not require regular zeros in order to construct the consumers; it
can therefore be carried out uniformly over any compact coefficient set.

The next proposition isolates exactly the two linearities needed for the conversion. Scaling
wealth scales demand, while averaging the coefficient $\sigma$ averages demand.

Two distinct structural properties drive the conversion below. Log-homogeneity of the
gauges makes Marshallian demand linear in wealth, allowing an endowment to be split into
income shares. Affinity of the explicitly decomposed potential in $\sigma$ makes demand
affine in the coefficient, allowing those shares to reconstruct the target parameter.
Neither property alone is sufficient.

\begin{proposition}[Affine-parameter and wealth-linear demand]\label{prop:affinedemand}
Let $P\subset\R^2$ be compact. The proof of Theorem~\ref{thm:parametricL} can be performed
with one cutoff and one constant $k$ uniformly for $\sigma\in P$. The resulting final
consumer has a family of fixed-parameter utilities $U^L_\sigma$ whose Marshallian demand is
of the form
\begin{equation}\label{eq:affinedemand}
x^L_\sigma(p,m)=m\bigl(d_0(p)+\sigma_1d_1(p)+\sigma_2d_2(p)\bigr),
\qquad p\gg0,\ m>0,
\end{equation}
for smooth vector fields $d_0,d_1,d_2$ independent of $\sigma$. Hence the demand is linear in wealth and affine in $\sigma$ on all of
$\R^L_{++}\times(0,\infty)$.
\end{proposition}

\begin{proof}
The core $\varphi^{[L]}_{L,\sigma}=-g^{[L]}_L$ is affine in $(\sigma_1,\sigma_2)$ by
\eqref{eq:g3} and \eqref{eq:gL}. Multiplication by a cutoff independent of $\sigma$
preserves that affine dependence. Since $P$ is compact, the derivative bounds in the gauge
construction are uniform on $P$, so a single $k$ works. The degree-one homogeneous utility
representation implies wealth-linear Marshallian demand; equivalently, this is explicit in
formula \eqref{eq:demandL}. Combining that wealth factor with the affine cutoff potential
gives \eqref{eq:affinedemand}.
\end{proof}

The next lemma contains the entire preference-to-endowment conversion. A consumer whose
demand is affine in a $d$-dimensional coefficient can be replicated by $d+1$ consumers at
fixed coefficient values; only the shares of the original endowment vary. In our
application $d=2$, so three fixed-preference consumers suffice.

\begin{lemma}[Barycentric preference-to-endowment conversion]\label{lem:barycentric}
Let $P\subset\R^d$ be a compact subset of the interior of the convex hull of affinely
independent points $t^0,\ldots,t^d$. Suppose a family of Marshallian demands satisfies
\[
x_t(p,m)=m\left(d_0(p)+\sum_{j=1}^dt_jd_j(p)\right).
\]
Let $\alpha_0(t),\ldots,\alpha_d(t)$ be the barycentric coordinates of $t$ and let
$\omega\in\R^L_+$ satisfy $p\cdot\omega>0$ for all $p\gg0$. Then all
$\alpha_r$ are uniformly positive on $P$. A consumer with demand $x_t$
and fixed endowment $\omega$ can be replaced, for each $t\in P$, by $d+1$ consumers with
fixed demands $x_{t^0},\ldots,x_{t^d}$ and endowments
\[
\omega^r_t=\alpha_r(t)\omega,\qquad r=0,\ldots,d.
\]
Their aggregate endowment is $\omega$, and their aggregate excess demand is exactly the
original consumer's excess demand at every strictly positive price. The replacement uses
$d+1$ fixed-preference consumers because an affine $d$-dimensional parameter is represented
by barycentric coordinates in a $d$-simplex; here $d=2$, so three consumers suffice.
\end{lemma}

\begin{proof}
Because $P$ is compactly contained in the interior of the simplex, continuity of the
barycentric coordinates gives $\min_{t\in P,r}\alpha_r(t)>0$. Since
$\sum_r\alpha_r(t)=1$ and $\sum_r\alpha_r(t)t^r=t$,
\[
\sum_{r=0}^d\omega^r_t=\omega.
\]
Put $m=p\cdot\omega$. The wealth of consumer $r$ is $\alpha_r(t)m$, and therefore
\begin{align*}
\sum_{r=0}^d x_{t^r}\bigl(p,p\cdot\omega^r_t\bigr)
&=\sum_{r=0}^d\alpha_r(t)x_{t^r}(p,m)\\
&=m\left(d_0(p)+\sum_{j=1}^d
       \left[\sum_{r=0}^d\alpha_r(t)t^r_j\right]d_j(p)\right)\\
&=x_t(p,m).
\end{align*}
Subtracting the identical aggregate endowment proves equality of excess demands.
\end{proof}

We can now eliminate preference variation completely. The single varying consumer of
Theorem~\ref{thm:parametricL} is replaced by three consumers with fixed preferences. Their
endowment shares vary, but their sum is always the original endowment, and aggregate excess
demand is unchanged exactly.

\begin{theorem}[Exact realization with fixed preferences and varying endowments]
\label{thm:endowmentreal}
For every $L\ge3$ there is a continuous family of $N=2L+2$ standard pure-exchange
economies $\{\widetilde\eco^{[L]}_\lambda\}_{\lambda\in A}$ with the following
properties:
\begin{enumerate}[label=\textup{(\alph*)},leftmargin=2.2em]
\item every preference relation is independent of $\lambda$;
\item only three individual endowments vary, their sum is constant, and the aggregate
resource vector of the economy is independent of $\lambda$;
\item for each $\lambda\in A$, the aggregate excess demand of
$\widetilde\eco^{[L]}_\lambda$ is identical at every $p\gg0$ to that of a
one-varying-consumer family supplied by Theorem~\ref{thm:parametricL} (constructed with the
same cutoff and $k$ on a larger compact coefficient set);
\item consequently, on $\Gamma_L$ the aggregate excess demand equals
$Z^{[L]}(\cdot;\lambda)$ exactly.
\end{enumerate}
The varying endowments return exactly to their initial values along every closed parameter
loop. The bound $2L+2$ is constructive and is not claimed to be minimal.
\end{theorem}

\begin{proof}
Let $s_{\max}=\sup_{\lambda\in A}|S(\lambda)|$ and choose $\kappa>2s_{\max}$. Consider the
three coefficient vectors
\[
\sigma^1=(\kappa,0),\qquad
\sigma^2=\left(-\frac\kappa2,\frac{\sqrt3\kappa}{2}\right),\qquad
\sigma^3=\left(-\frac\kappa2,-\frac{\sqrt3\kappa}{2}\right).
\]
Their equilateral triangle $P$ has inradius $\kappa/2$ and therefore contains $S(A)$ in
its interior. For $\sigma=(\sigma_1,\sigma_2)$ its barycentric coordinates are
\begin{align}
\alpha_1(\sigma)&=\frac13+\frac{2\sigma_1}{3\kappa},\nonumber\\
\alpha_2(\sigma)&=\frac13-\frac{\sigma_1}{3\kappa}
                  +\frac{\sigma_2}{\sqrt3\kappa},\label{eq:baryalpha}\\
\alpha_3(\sigma)&=\frac13-\frac{\sigma_1}{3\kappa}
                  -\frac{\sigma_2}{\sqrt3\kappa}.\nonumber
\end{align}
For $|\sigma|\le s_{\max}$,
\[
\alpha_r(\sigma)\ge
\frac13-\frac{2s_{\max}}{3\kappa}>0,
\]
so all three consumers have positive wealth at positive prices.

Apply Proposition~\ref{prop:affinedemand} on the whole triangle $P$, and consider the
one-varying-consumer economy at the coefficient $\sigma=S(\lambda)$. Its last consumer has
endowment $\omega^L=k\mathbf e_L$. Replace that consumer by three consumers with the fixed
preferences represented by
$U^L_{\sigma^1},U^L_{\sigma^2},U^L_{\sigma^3}$ and with endowments
\[
\omega^{L,r}_\lambda=\alpha_r(S(\lambda))\,k\mathbf e_L,
\qquad r=1,2,3.
\]
Lemma~\ref{lem:barycentric} shows that their aggregate demand and aggregate excess demand
coincide exactly with those of the removed consumer at every $p\gg0$. All other consumers
are unchanged. Hence the two economic families have identical aggregate excess demand
globally for each $\lambda$, while the three new endowments sum to $k\mathbf e_L$.
Replacing one consumer by three raises the count from $2L$ to $2L+2$. Continuity and the
return property follow from continuity and periodicity of $S$ and of the affine functions
\eqref{eq:baryalpha}.
\end{proof}

\begin{remark}[A visible redistribution loop]\label{rem:shares}
The effective coefficient parameter is two-dimensional, which explains the use of three
replacement consumers. With the standing generator $r^2=1.03$, one has
$|S(\lambda)|=r(r^2-1)\approx0.030447$. Choosing $\kappa=0.11$ in
\eqref{eq:baryalpha}, each barycentric share varies over
\[
\alpha_r(\theta)\in
\left[\frac13-\frac{2|S|}{3\kappa},
      \frac13+\frac{2|S|}{3\kappa}\right]
\approx[0.149,0.518].
\]
Thus the illustrative endowment loop is not infinitesimal: each replacement consumer's
share of the original vertex endowment ranges from about $14.9\%$ to $51.8\%$. These
numbers illustrate the chosen generator and triangle; they are not intrinsic bounds of the
theorem.
\end{remark}

\section{Real monodromy under endowment and preference loops}\label{sec:main}

Everything needed for the economic result is now in place. The formal construction supplied
the equilibrium-price monodromy; the parametric realization converted it into a continuous
family with one varying preference; and Section~\ref{sec:endowment} converted that family
into one with fixed preferences and varying endowments. We can therefore state the result
directly as a closed economic experiment: the economy returns exactly to its initial
primitives, but a regular equilibrium price followed continuously around the loop can
return as a different equilibrium price.

Moreover, the phenomenon does not require changing preferences. It can be generated by
redistributing a fixed aggregate endowment among three consumers whose preferences remain
unchanged. The theorem gives two exact microfoundations of the same aggregate experiment.

\begin{theorem}[Endowment and preference realizations of real monodromy]\label{thm:main}
For every integer $L\ge3$ there exist a compact normalized positive-price region $K_L$ and
two continuous one-parameter periodic families of standard finite-agent pure-exchange
economies with $L$ goods,
\[
\{\eco^{\mathrm{end},[L]}_\theta\}_{\theta\in S^1},
\qquad
\{\eco^{\mathrm{pref},[L]}_\theta\}_{\theta\in S^1},
\]
with the following properties.
\begin{enumerate}[label=\textup{(\alph*)},leftmargin=2.2em]
\item In the endowment family, all preference relations are fixed, only three individual
endowments vary, and their sum---hence aggregate resources---is constant; the family has at
most $2L+2$ consumers. \item In the preference family, all endowments are fixed and only one consumer's preference
relation varies; the family has at most $2L$ consumers.
\item For every parameter value, the two families have identical aggregate excess demand at
all strictly positive prices.
\item Each restricted equilibrium-price fibre in $K_L$ consists of exactly three regular
prices with reduced indices $+1,+1,-1$. The restricted covering is the disjoint union of a
fixed sheet and a connected double cover.
\item One circuit exchanges the two positive-index equilibrium prices and fixes the
negative-index price.
\item After one circuit every primitive of the relevant family returns exactly to its
initial value.
\end{enumerate}
Thus one periodic parameter suffices. No assertion is made about the number or regularity of
normalized equilibrium prices outside $K_L$.
\end{theorem}

The endowment realization has an additional presentational advantage: its loop lives in an
ordinary finite-dimensional Euclidean space, so no topology on ordinal preference relations
and no choice of cardinal utility representatives is needed to state it. The phrase
``fixed sheet and connected double cover'' has a simple economic meaning: one of the three
equilibrium prices can be identified globally throughout the experiment, whereas the other
two can only be distinguished locally; after one circuit their identities are reversed.

Crucially, the permutation is not produced by failing to close the economic experiment. At
the end of the circuit the economy is exactly the initial economy---same preferences and
same endowments---yet the continuation of one regular equilibrium price ends at another
equilibrium price of that economy.

The three-branch geometry constructed formally in Section~\ref{sec:formal} is therefore
realized exactly by both economic families of Theorem~\ref{thm:main}, in particular by the
fixed-preference endowment-redistribution family.

\begin{proof}
Let $A$ be the annulus of Theorem~\ref{thm:formal} and let
$\gamma\colon S^1\to A$ be its standard generator. Construct the preference family with the
uniform coefficient triangle used in Theorem~\ref{thm:endowmentreal}; on
$\sigma=S(\lambda)$ it is a realization of Theorem~\ref{thm:parametricL}.
Theorem~\ref{thm:endowmentreal} converts it into the endowment family and shows that the two
aggregate excess-demand fields are globally identical. Restricting both families to
$\gamma$ gives the stated loops. On the cone
containing $K_L$, both therefore coincide with the formal field $Z^{[L]}$. Theorem~\ref{thm:formal} and Lemma~\ref{lem:swap} transfer the three-sheeted covering, indices, and
transposition to both economic families. Periodicity gives the exact return of all
primitives.
\end{proof}

The annulus is only topological scaffolding; it does not mean that two independently
varying economic parameters are required. The economic experiment of the main theorem is
one-dimensional, obtained by restricting the annular family to a single closed generator.

\begin{proposition}[Annular extension]\label{prop:annular}
Both loops in Theorem~\ref{thm:main} are restrictions of continuous families over the
compact annulus $A$. For each family the restricted equilibrium-price correspondence over
$K_L\times A$ is the disjoint union of a trivial sheet and a connected double cover. The endowment realization has fixed preferences and a constant aggregate resource vector;
the preference realization has fixed endowments and one varying preference. Their aggregate
excess-demand fields are identical for every $\lambda\in A$ and every $p\gg0$.
\end{proposition}

\begin{proof}
Apply Theorems~\ref{thm:parametricL} and~\ref{thm:endowmentreal} over the annulus. Equality
of the aggregate fields follows from Lemma~\ref{lem:barycentric}; the covering statement
then follows from Theorem~\ref{thm:formal}.
\end{proof}

One further point must be checked. A nonconstant formal parameterization would not be
economically meaningful if different parameter values represented the same primitives. In
our construction the parameter is faithfully recorded by the varying preference in one
realization and by the three endowment shares in the other.

\begin{proposition}[Faithfulness of the two parameter loops]\label{prop:faithful}
The annular preference map and the annular endowment map are injective. In particular,
their restrictions to the standard generator are nonconstant embedded circles: the first in
the represented preference-family space of Definition~\ref{def:reptopology}, the second in
the finite-dimensional Euclidean endowment space.
\end{proposition}

\begin{proof}
For the preference family, the preference relation of the varying consumer determines its
unique demand on every strictly positive budget. Formula \eqref{eq:demandL}, with fixed
$k$ and logarithmic baseline, therefore determines the gradient of the cutoff potential. At
the normalized price $q=\one$, where the cutoff is flat, the last core satisfies
\[
\partial_u\varphi^L_\lambda(0,0,0)=-c\sigma_1(\lambda),\qquad
\partial_v\varphi^L_\lambda(0,0,0)=-c\sigma_2(\lambda),
\]
Thus demand determines $S(\lambda)$. The map
$S(\lambda)=\lambda(|\lambda|^2-1)$ is injective on $A$, because it preserves the argument
and its radial modulus $r(r^2-1)$ is strictly increasing for $r>1$. Hence distinct
parameters induce distinct preference relations.

For the endowment family, the three varying endowments determine their barycentric
coefficients. The map $\sigma\mapsto(\alpha_1(\sigma),\alpha_2(\sigma),\alpha_3(\sigma))$
is affine and injective, and $S$ is injective on $A$. The endowment map is therefore injective. The relevant target spaces are Hausdorff, so a
continuous injection from the compact circle is an embedding.
\end{proof}

The circular experiment is only the simplest way to display the phenomenon. Any parameter
space that maps continuously into the annulus inherits the same monodromy; what matters is
the winding of the induced loop.

\begin{corollary}[Parameter-space pullback]\label{cor:pullback}
Let $B$ be a connected, locally path-connected parameter space and let $f\colon B\to A$ be
continuous. Pulling back either annular family gives a continuous pure-exchange family over
$B$. For a loop $\beta$ in $B$, the two mobile prices are exchanged precisely when
$f_*[\beta]\in\pi_1(A)\cong\Z$ has odd winding number; they are fixed when it is even.
Thus neither the annulus nor two economic parameters are intrinsic to the phenomenon.
\end{corollary}

\begin{proof}
The restricted equilibrium covering of either pulled-back family is the pullback of the
annular covering. Its monodromy is the composition of $f_*$ with the transposition assigned
to a generator of $\pi_1(A)$.
\end{proof}

\begin{remark}[Interpretation of the two realizations]\label{rem:onetaste}
The one-varying-preference realization is a concentration result, not a necessity statement:
several preferences could vary, but they need not. The endowment realization shows that the
same monodromy can be generated with every preference fixed. The bounds $2L$ and $2L+2$ are
constructive upper bounds only.
\end{remark}

\begin{remark}[Equilibria outside the target region]\label{rem:outside}
The theorem is global with respect to the equivalence of the two economic realizations, but
local with respect to the formal field used to control the equilibrium structure. The
preference and endowment families have identical aggregate excess demand at every positive
price, whereas their equality with the formal field is imposed only on the target cone. We
therefore know exactly what happens to the three branches in $K_L$, but do not claim that
they are the only equilibria of the economy or that equilibria outside the controlled
region are everywhere regular. Whether one can construct such a loop while
controlling the complete equilibrium set globally remains open.
A related stronger question, in equilibrium-manifold language, is whether nontrivial
monodromy can occur for the full finite covering obtained by restricting the natural
projection to regular economies.
\end{remark}

\section{Structural stability of the monodromy}\label{sec:robust}

The monodromy established above is not tied to the particular formulas of the
construction. It persists under sufficiently small $C^1$ perturbations of the reduced
aggregate excess-demand system. If an economic family is sufficiently close to the
constructed family---both in aggregate excess demand and in its price derivatives---the
three regular branches persist and the same two branches are exchanged. Thus the result
holds on a $C^1$ neighborhood of the constructed reduced system.

Two kinds of control are required. Closeness in function values keeps equilibria confined
near the original equilibrium locus and prevents new zeros from entering the target
region. Closeness of price derivatives preserves nonsingularity of the equilibrium
Jacobians. This is why the natural robustness topology is $C^1$, rather than merely $C^0$.

\subsection{Robustness at the aggregate level}

Fix $L\ge3$. Let $F=F^{[L]}$ and let
$M=F^{-1}(0)=\eco_{K_L}\subset\operatorname{int}K_L\times A$. Since $M$ is compact and $D_xF$ is invertible on $M$, compactness gives a uniform regularity margin:
\begin{equation}\label{eq:mu}
\mu=\min_{(x,\lambda)\in M}\sigma_{\min}(D_xF(x,\lambda))>0.
\end{equation}

The regularity margin is needed only near the equilibrium locus: the reduced Jacobian may
be singular at nonequilibrium prices elsewhere in $K_L$. Choose a relatively open
neighborhood $N$ of $M$ with
$\overline N\subset\operatorname{int}K_L\times A$ and
$\sigma_{\min}(D_xF)\ge\mu/2$ on $\overline N$. A second margin isolates the equilibrium locus from the rest of the target region. Since $F$ has no zero outside $N$,
\begin{equation}\label{eq:alphaeta}
\alpha=\min_{(K_L\times A)\setminus N}|F|>0,\qquad
\varepsilon_L=\min\{\alpha/2,\mu/4\}>0.
\end{equation}

These two margins exclude exactly the events that could destroy the covering during a
small perturbation: no zero can appear away from the original branches or cross the
boundary, and no continued zero can become singular. The homotopy parameter $s$ below is
only a mathematical deformation from the original to the perturbed system; it is not an
economic parameter.

\begin{lemma}[Localization and uniform regularity along a homotopy]\label{lem:isolation}
Let $\widetilde F$ be $C^1$ near $K_L\times A$ and suppose
$\|\widetilde F-F\|_{C^1(K_L\times A)}<\varepsilon_L$. For
$F_s=(1-s)F+s\widetilde F$, every zero of $F_s$ lies in $N$, none lies on
$\partial K_L\times A$, and
$\sigma_{\min}(D_xF_s)\ge\mu/4$ at every zero.
\end{lemma}

\begin{proof}
Outside $N$, $|F_s|\ge\alpha-\varepsilon_L\ge\alpha/2$. On $\overline N$, the perturbation
inequality for singular values gives
$\sigma_{\min}(D_xF_s)\ge\mu/2-\varepsilon_L\ge\mu/4$.
\end{proof}

A small perturbation may move the numerical equilibrium prices. What it cannot do, under
the bounds above, is change how the three branches are connected over the parameter loop.
The permutation cannot change without a collision, a boundary crossing, or a loss of
regularity, and Lemma~\ref{lem:isolation} excludes all three events.

\begin{proposition}[Robustness of the monodromy in arbitrary dimension]\label{prop:robust}
If $\|\widetilde F-F^{[L]}\|_{C^1(K_L\times A)}<\varepsilon_L$, then the zero locus of
$\widetilde F$ in $K_L\times A$ is a three-sheeted covering isomorphic to
$\eco_{K_L}\to A$. It has one fixed index-$-1$ sheet and a connected index-$+1$ double
cover; its monodromy along a generator is the transposition of the two sheets of the double
cover. Hence every family of standard $L$-good economies whose reduced equilibrium-price
system is sufficiently $C^1$-close to $F^{[L]}$ on $K_L\times A$ exhibits the same
restricted monodromy. \end{proposition}

\begin{proof}
Consider
\[
\mathscr M=\{(x,\lambda,s)\in K_L\times A\times[0,1]:F_s(x,\lambda)=0\}.
\]
Lemma~\ref{lem:isolation} places $\mathscr M$ in
$\operatorname{int}K_L\times A\times[0,1]$ and makes $D_xF_s$ invertible throughout.
Thus projection $\mathscr M\to A\times[0,1]$ is a proper local homeomorphism, hence a
finite covering. Its degree is three because the slice at $s=0$ is the covering of
Theorem~\ref{thm:formal}. Since $[0,1]$ is contractible, the restrictions at $s=0$ and $s=1$ are isomorphic coverings,
so their monodromy representations have the same cycle type. The determinant sign cannot
change along the homotopy, and therefore the indices remain $+1,+1,-1$.
\end{proof}

\subsection{From economic perturbations to aggregate perturbations}

Proposition~\ref{prop:robust} is stated at the aggregate level. To make the robustness claim
economically meaningful, we must connect small changes in individual demands and
endowments to small changes in aggregate excess demand. Aggregation is a finite sum, so on
a compact price--wealth--parameter region small $C^1$ changes in a finite smooth consumer
block produce a small $C^1$ change in the reduced aggregate system.

We perturb only the explicit smooth consumer block. The fixed Debreu residual consumers are
kept unchanged and enter through their known aggregate contribution, denoted by the common
background field $H$ below. Thus no differentiability assumption on their individual
demands is required.

\begin{lemma}[Perturbations of a smooth consumer block]\label{lem:demandbridge}
Let $\mathcal P_L=\{(1+x_1,\ldots,1+x_{L-1},1):x\in K_L\}$. Suppose the reduced systems of
two parameterized economies can be written on $K_L\times A$ as
\[
F=H+\sum_{i\in J}F_i,
\qquad
\widetilde F=H+\sum_{i\in J}\widetilde F_i,
\]
where $H$ is a common $C^1$ {reduced} background field and $J$ is a finite set of consumers
whose primitives may differ. For $i\in J$, suppose the Marshallian demands
$x^i(p,m,\lambda)$ and $\widetilde x^i(p,m,\lambda)$ are $C^1$ {on a neighbourhood of the compact set}
$\mathcal P_L\times[m_-,m_+]\times A$, and the parameterized endowments
$\omega^i,\widetilde\omega^i\colon A\to\R^L_+$ are $C^1$. {Assume moreover that, for every
$i\in J$, $p\in\mathcal P_L$, and $\lambda\in A$, both induced wealth levels
$p\cdot\omega^i(\lambda)$ and $p\cdot\widetilde\omega^i(\lambda)$ belong to
$[m_-,m_+]\Subset(0,\infty)$.} Then, for every $\varepsilon>0$, there exists $\delta>0$
such that
\[
\sum_{i\in J}\left(
\|\widetilde x^i-x^i\|_{C^1(\mathcal P_L\times[m_-,m_+]\times A)}
+\|\widetilde\omega^i-\omega^i\|_{C^1(A)}\right)<\delta
\]
implies
\[
\|\widetilde F-F\|_{C^1(K_L\times A)}<\varepsilon.
\]
Thus the reduced system depends continuously in the $C^1$ topology on the demands and
endowments of the perturbed finite consumer block. No differentiability assumption is
required for consumers whose aggregate contribution is contained in the unchanged block
$H$.
\end{lemma}

\begin{proof}
For $i\in J$, write
\[
z^i(p,\lambda)=x^i\bigl(p,p\cdot\omega^i(\lambda),\lambda\bigr)-\omega^i(\lambda),
\]
and define $\widetilde z^i$ analogously. The maps
\[
(p,\lambda)\longmapsto
\bigl(p,p\cdot\omega^i(\lambda),\lambda\bigr)
\]
depend continuously in the $C^1$ topology on $\omega^i$. Since the derivatives of each
reference demand $x^i$ are uniformly continuous on the compact
price--wealth--parameter set, the chain rule implies that
\[
\|\widetilde x^i-x^i\|_{C^1}\to0,
\qquad
\|\widetilde\omega^i-\omega^i\|_{C^1(A)}\to0
\]
entails
\[
\|\widetilde z^i-z^i\|_{C^1(\mathcal P_L\times A)}\to0.
\]
Finite summation over $J$ and restriction to the first $L-1$ components on the normalized
price slice preserve $C^1$ convergence. Finally, the common block cancels identically:
\[
\widetilde F-F=\sum_{i\in J}(\widetilde F_i-F_i).
\]
This proves the stated $C^1$ continuity.
\end{proof}

{
In the economies constructed above, the unchanged residual consumers contribute on the
target cone the explicitly known full aggregate excess-demand field $-kR_L$. Restricting
its first $L-1$ components to the normalized price slice gives the common reduced background
field $H$ in Lemma~\ref{lem:demandbridge}, which is $C^\infty$ on $K_L\times A$. Thus the
reference reduced system is $C^\infty$ there even though differentiability of each residual
consumer's Debreu demand is neither assumed nor needed.
}

For the reference family, compactness and strict positivity of wealth allow the interval
$[m_-,m_+]$ in Lemma~\ref{lem:demandbridge} to be chosen with slack; sufficiently small
endowment perturbations therefore preserve the wealth-range hypothesis.

Combining aggregate structural stability with continuity of aggregation gives the
economically relevant robustness statement.

\begin{corollary}[Economic perturbations preserve the transposition]
\label{cor:econrobust}
Every sufficiently small $C^1$ perturbation of a finite smooth consumer block and its
parameterized endowments, while a common aggregate background block is held fixed, induces
the same restricted monodromy. This includes simultaneous changes in several consumers'
preferences whenever those changes generate sufficiently small $C^1$ changes in their
Marshallian demands on the relevant compact price--wealth set.
\end{corollary}

\begin{proof}
Apply Lemma~\ref{lem:demandbridge} with $\varepsilon=\varepsilon_L$. For sufficiently small
$C^1$ perturbations of the demands and endowments in the finite block $J$, one has
\[
\|\widetilde F-F\|_{C^1(K_L\times A)}<\varepsilon_L.
\]
Proposition~\ref{prop:robust} then gives the same restricted monodromy.
\end{proof}

\section{Why two goods are not enough}\label{sec:twogood}

The examples of Toda and Walsh \cite{todawalsh17} show that explicit two-good
economies may already have three or more equilibria. The obstruction to monodromy is
therefore not a lack of multiplicity. It is the one-dimensional geometry of normalized
prices.

\subsection{The ordering obstruction}

Suppose that at every parameter value the equilibrium prices can be written from smallest
to largest. As long as the equilibria remain distinct and regular, their ordering from
lowest to highest is preserved under continuation. A change in rank would require two
equilibrium branches to meet, hence a collision. Equivalently, the rank
\[
r(p,\lambda)=1+\#\{q\in\operatorname{Eq}(\lambda):q<p\}
\]
is preserved under regular continuation. The following elementary covering theorem
formalizes this observation.

\begin{theorem}[Ordered finite coverings are trivial]\label{thm:orderedcover}
Let $A$ be connected and locally path connected, let $I\subset\R$ be an interval, and let
$E\subset I\times A$.  Suppose that the restriction of the projection
\[
\pi\colon E\longrightarrow A,\qquad \pi(x,\lambda)=\lambda,
\]
is a finite covering of degree $m$.  Then the covering is trivial.  More precisely, there
exist continuous functions
\[
s_{1},\ldots,s_{m}\colon A\longrightarrow I,
\qquad
s_{1}(\lambda)<\cdots<s_{m}(\lambda)
\]
such that
\[
E=\bigsqcup_{j=1}^{m}\{(s_{j}(\lambda),\lambda):\lambda\in A\}.
\]
Consequently the monodromy representation
\[
\rho\colon\pi_{1}(A,\lambda_{0})\longrightarrow \mathfrak S_{m}
\]
is trivial.
\end{theorem}

\begin{proof}
For each $\lambda\in A$, write the fibre in increasing order as
\[
\pi^{-1}(\lambda)=\{(x_{1}(\lambda),\lambda),\ldots,(x_{m}(\lambda),\lambda)\},
\qquad
x_{1}(\lambda)<\cdots<x_{m}(\lambda).
\]
It remains to prove that the order statistics $x_{j}$ are continuous.  Fix
$\lambda_{0}\in A$.  Since $\pi$ is a covering, there is an evenly covered neighbourhood
$U$ of $\lambda_{0}$ and continuous local sections
\[
\sigma_{1},\ldots,\sigma_{m}\colon U\longrightarrow E
\]
whose images are the sheets over $U$.  Write
\[
\sigma_{k}(\lambda)=(y_{k}(\lambda),\lambda).
\]
After relabelling, assume
$y_{1}(\lambda_{0})<\cdots<y_{m}(\lambda_{0})$.  Choose pairwise disjoint open intervals
$J_{1},\ldots,J_{m}\subset I$ with
$y_{j}(\lambda_{0})\in J_{j}$ and every point of $J_{j}$ smaller than every point of
$J_{j+1}$.  By continuity, after shrinking $U$ we have $y_{j}(U)\subset J_{j}$ for every
$j$.  Hence the order of the local sheets cannot change on $U$, and the $j$th order
statistic coincides there with $y_{j}$.  Thus every $x_{j}$ is locally continuous and
therefore continuous on $A$.  Their graphs are disjoint sections whose union is $E$.
A loop preserves each global section, so its action on the fibre is the identity.
\end{proof}

With two commodities, normalization leaves only one independent relative price. The
positive normalized price space is therefore an interval, and every finite regular fibre
inherits its ordinary left-to-right order.

\begin{corollary}[Two-good impossibility]\label{cor:twogood}
Consider a continuous family of two-good pure-exchange economies indexed by a connected,
locally path-connected, locally compact Hausdorff parameter space $A$. Normalize prices
either by $p_2=1$, so that the relative price lies in $(0,\infty)$, or by $p_1+p_2=1$, so
that it lies in $(0,1)$. If the family is a restricted regular equilibrium-price covering
over a compact positive-price region in the sense of Definition~\ref{def:restrictedcover},
then its real monodromy is trivial.

If individual demands are single valued, the same conclusion holds for the full
price-allocation equilibrium correspondence, because the allocation is determined by the
price.
\end{corollary}

\begin{proof}
The normalized positive-price space is an interval, so Theorem~\ref{thm:orderedcover}
applies to the equilibrium-price locus.  Under single-valued demand, projection from a
full equilibrium to its price is one-to-one on each fibre.
\end{proof}

Combining the ordering obstruction with the constructions above gives the sharp commodity
threshold.

\begin{corollary}[Sharp threshold]\label{cor:minimalgoods}
Within the class of restricted regular equilibrium-price coverings considered here,
nontrivial real monodromy exists exactly for economies with $L\ge3$ goods.
\end{corollary}

\begin{proof}
With one good there is no relative-price variable. With two goods,
Corollary~\ref{cor:twogood} gives trivial monodromy. For every $L\ge3$,
Theorem~\ref{thm:main} supplies a one-parameter periodic example.
\end{proof}

\subsection{Multiplicity without monodromy: Toda--Walsh}

Toda and Walsh's multiple-equilibrium examples \cite{todawalsh17} illustrate the distinction
between multiplicity and monodromy: whenever their real equilibria form a finite restricted
regular covering, the ordering theorem applies.

\begin{remark}[Real versus complex monodromy]\label{rem:realcomplex}
In Toda and Walsh's quadratic example the two non-unit equilibria are
\[
p_{\pm}(\alpha)=\frac{1-2k\pm\sqrt{1-4k}}{2k},
\qquad k=2\alpha(1-\alpha),\qquad p_-p_+=1,
\]
where $\alpha$ is the preference weight in their symmetric quadratic utility specification.
The real three-equilibrium region is $1-8\alpha(1-\alpha)>0$, and on each connected
component $p_-(\alpha)<1<p_+(\alpha)$. This order gives a global real labelling. If
$\alpha$ is complexified, a loop around a zero of $1-8\alpha(1-\alpha)$ exchanges the two
algebraic branches through the square root. That is complex algebraic monodromy, not real
economic monodromy: no corresponding loop stays in the real parameter region while keeping
the roots real, positive, and regular. The contrast explains why the adjective ``real'' is substantive here. In the three-good construction, two normalized price coordinates provide a real plane in which branches can move around a degeneracy without colliding. In the two-good case there is only a real line: one cannot pass around a collision point; one can only cross it.
\end{remark}

Exact factorizations for the remaining two-good benchmarks are collected in
Appendix~\ref{app:twogood}.

\begin{remark}[Why a passive third good is insufficient]\label{rem:passivegood}
Adding a third good through an independent equation, for example
\[
z_{\mathrm{TW}}(p;\lambda)=0,
\qquad q=q_{0},
\]
does not create monodromy.  The equilibrium points remain on the ordered line
$\{(p,q_{0})\}$, so the preceding obstruction survives.  A genuine three-good extension
of a Toda--Walsh class must couple the two normalized prices and create a two-dimensional
solution geometry around a discriminant.  \end{remark}

The sharp threshold is therefore geometric rather than multiplicity-based: two goods leave
one normalized price dimension, where regular equilibria are globally rankable; three goods
provide two normalized price dimensions, the minimum needed for real equilibrium branches
to wind around one another.

\section{Monodromy with prescribed equilibrium prices}\label{sec:anchors}

The particular equilibrium prices used in Section~\ref{sec:formal} were chosen only for
convenience. The next result shows that their location is inessential. Given any three
distinct positive price rays, we can construct a new pure-exchange family whose restricted
equilibrium-price fibre at the base parameter consists exactly of those three rays, with two
exchanged by continuation around the parameter loop and the third fixed.

This result allows the initial fibre to be chosen, for example, from equilibrium prices
appearing in published general-equilibrium examples. The resulting monodromic family is
nevertheless a new economy: the construction does not modify or continue the source economy
from which those prices may have been taken.

The possibility of prescribing the equilibrium-price set of a single exchange economy is
classical. Mas-Colell shows that any nonempty compact set of normalized prices can arise as
the equilibrium-price set of an exchange economy and, for regular configurations with more
than two commodities, that the fixed-point-index condition essentially exhausts the
remaining restrictions \cite{mascolell77}. Our result concerns a different question. We
prescribe three prices in the initial fibre of a continuous family and, in addition,
prescribe their continuation around a loop: two branches are exchanged and the third is
fixed. The result concerns the restricted price region under construction, not the complete
equilibrium set of the economy.

The proof uses the following parametric version of Mantel's constructive realization. It is
a dimension-$L$ extension of the analytic argument proved in Appendix~\ref{app:analytic};
unlike the special realization used for the main theorem, it does not concentrate all
parameter dependence in a single consumer.

\begin{proposition}[Generic parametric realization]\label{prop:dimfree}
Let $L\ge2$, let $\Lambda\Subset W\subset\R^{d}$ with $W$ open, and let $\Gamma$ be a closed
cone of strictly positive prices whose image in the normalized simplex is compactly
contained in its interior. Suppose
\[
Z\colon \R^{L}_{++}\times W\longrightarrow\R^{L}
\]
is jointly $C^{\infty}$, homogeneous of degree zero in prices, and satisfies Walras' law
$p\cdot Z(p;\lambda)=0$. Then there is a family of finite-consumer pure-exchange economies,
with parameter-independent endowments, whose aggregate excess demand equals
$Z(\,\cdot\,;\lambda)$ on $\Gamma$ for every $\lambda\in\Lambda$. The family can be chosen
with at most $2L$ consumers. The first $L$ consumers admit utility representations jointly
continuous on $\R^L_+\times\Lambda$ and jointly $C^\infty$, strongly monotone, and strictly
quasiconcave on $\R^L_{++}\times W$; the remaining $L$ consumers have fixed continuous,
monotone, strictly convex preferences. Only the first block may depend on $\lambda$.
\end{proposition}

\begin{proof}
Choose a smooth homogeneous cutoff $\chi$ that equals one on a neighborhood of $\Gamma$ and
has compact support after normalization, and set
$\varphi^i_\lambda(p)=\chi(p)Z_i(p;\lambda)$. On $\Gamma$, differentiation of Walras' law
gives
\begin{equation}\label{eq:dimfree-potential}
\sum_{i=1}^{L}p_i\nabla\varphi^i_\lambda(p)=-Z(p;\lambda).
\end{equation}
For $k>0$ define
\[
\psi^i_\lambda(q)=\frac{\varphi^i_\lambda(q)}{k}
-\frac1L\sum_{j=1}^{L}\log q_j.
\]
The cutoff gives uniform bounds on the scaled first and second derivatives of
$\varphi^i_\lambda$. Hence, for finite constants $C,C_1$ independent of $i$ and $\lambda$,
\[
D^2\psi^i_\lambda(q)\succeq
\left(\frac1L-\frac{C}{k}\right)\widehat q^{-2},
\qquad
-q_j\partial_j\psi^i_\lambda(q)\ge \frac1L-\frac{C_1}{k}.
\]
Choose $k>L\max\{C,C_1\}$. The global inverse-demand, utility-construction, and boundary
arguments of Appendix~\ref{app:analytic} use only these estimates, homogeneity, and the
radial identity $q\cdot\nabla\psi^i_\lambda=-1$; they therefore apply verbatim in dimension
$L$, with $1/3$ replaced by $1/L$. Giving consumer $i$ the fixed endowment
$k\mathbf e_i$ yields demand
\[
x^i_\lambda(p,kp_i)=
\frac{kp_i}{L}\widehat p^{-1}\one-p_i\nabla\varphi^i_\lambda(p).
\]
Thus the first block has aggregate excess demand
\begin{equation}\label{eq:dimfree-block}
Z(p;\lambda)+kR_L(p),\qquad
R_L(p)=\frac{\one\cdot p}{L}\widehat p^{-1}\one-\one,
\end{equation}
on $\Gamma$, by \eqref{eq:dimfree-potential}. The field $R_L$ is Walrasian, homogeneous of
degree zero, and independent of $\lambda$. One application of Debreu's realization theorem
to $-kR_L$ therefore supplies a fixed residual block of at most $L$ consumers. This gives
the stated family and the bound $2L$.
\end{proof}

\begin{theorem}[Prescribed three-point equilibrium-price fibre]\label{thm:anchored}
Let $L\ge3$, and let $q_{+},q_{-},q_{0}\in\R^{L}_{++}$ be three distinct price rays.
There exist a parameter annulus $A_{*}\subset\C$, a compact normalized price region
$K_{*}$, and a continuous family of finite-consumer pure-exchange economies with fixed
endowments such that:
\begin{enumerate}[label=\textup{(\alph*)},leftmargin=2.2em]
\item over the base point $\lambda_{0}=1$, the restricted equilibrium-price fibre is
exactly $\{q_{+},q_{-},q_{0}\}$;
\item the three equilibrium prices are regular and have indices $+1,+1,-1$, respectively;
\item lifting the generator of $\pi_{1}(A_{*})$ exchanges $q_{+}$ and $q_{-}$ and fixes
$q_{0}$.
\end{enumerate}
The family may be chosen with at most $2L$ consumers, as a constructive upper bound.
\end{theorem}

\begin{proof}
Normalize the three rays by $p_L=1$ and use log-relative-price coordinates
\[
x(p)=\left(\log\frac{p_1}{p_L},\ldots,\log\frac{p_{L-1}}{p_L}\right)\in\R^{L-1}.
\]
Write $x_\nu=x(q_\nu)$ for $\nu\in\{+,-,0\}$ and
$m=(x_++x_-)/2$. Since $x_+\ne x_-$, choose a real linear isomorphism
$T\colon\R^{L-1}\to\R^{L-1}$ such that
\[
T(x_+-m)=(1,0,\ldots,0),\qquad
T(x_--m)=(-1,0,\ldots,0),
\]
and
\[
T(x_0-m)=(a,b,0,\ldots,0),\qquad a^2+b^2\ne1.
\]
Such a $T$ exists because, after sending the line through $x_+-x_-$ to the first coordinate
axis, a transverse coordinate can be rescaled to avoid the single forbidden radius; if
$x_0-m$ lies on that line, distinctness already gives $a\ne\pm1$. Put
$z=y_1+iy_2$, $\tau=(y_3,\ldots,y_{L-1})$, and $z_0=a+ib$.

For $\lambda\in\C$ define
\[
G(y;\lambda)=\left(
\re[(z^2-\lambda)(\bar z-\bar z_0)],
\im[(z^2-\lambda)(\bar z-\bar z_0)],
-\tau
\right).
\]
Choose $\varepsilon>0$ so small that
$|z_0|^2\notin[1-\varepsilon,1+\varepsilon]$ and set
$A_* = \{\lambda\in\C:1-\varepsilon\le |\lambda|\le1+\varepsilon\}$.
For every $\lambda\in A_*$ the zeros of $G$ are exactly
$(\sqrt\lambda,0)$, $(-\sqrt\lambda,0)$, and $(z_0,0)$; they are simple, the first two are
exchanged around a generator of the annulus, and the third is fixed. The planar factor is
locally holomorphic at the mobile zeros and anti-holomorphic at the fixed zero. Since the
transverse block in $-D_yG$ is $I_{L-3}$, the corresponding reduced indices are
$+1,+1,-1$ in every dimension.

Pull the system back by
\[
\widetilde F(x;\lambda)=T^{-1}G(T(x-m);\lambda)
\]
and complete it to a Walrasian field by setting, for $j<L$,
\[
Z_j^*(p;\lambda)=\widetilde F_j(x(p);\lambda),\qquad
Z_L^*(p;\lambda)=-\sum_{j=1}^{L-1}\frac{p_j}{p_L}\widetilde F_j(x(p);\lambda).
\]
Then $Z^*$ is smooth, homogeneous of degree zero, and satisfies Walras' law identically;
on $p_L=1$ its reduced zero system is $\widetilde F=0$. Conjugation by $T$ preserves the
determinant of the reduced linearization, while the log-price map has positive Jacobian
determinant, so the index signs found above are unchanged in the original normalized price
coordinates. The zero locus over $A_*$ is compact and contains exactly three points in each
fibre. Choose a compact neighborhood $K_*$ containing it in its interior and with no zero
on its boundary. Proposition~\ref{prop:dimfree} realizes $Z^*$ exactly on the cone over
$K_*$ with at most $2L$ consumers and fixed endowments. The restricted covering and its
monodromy therefore transfer to the resulting pure-exchange family.
\end{proof}

\begin{remark}[Scope of the prescribed-price construction]\label{rem:anchor-scope}
Theorem~\ref{thm:anchored} constructs a new family having the prescribed three price rays
in its initial restricted fibre. A price ray has no index independently of an excess-demand
system. In the literature applications of Appendix~\ref{app:anchors}, we select source
equilibria that already have the corresponding local index pattern, so the new family
matches both the selected rays and their local index labels. No other primitives or
equilibrium properties of the source economy are imposed on the constructed family.
\end{remark}

Appendix~\ref{app:anchors} applies the theorem to selected four-good equilibrium-price
fibres from the pure-exchange example of Gauthier--Kehoe--Quintin \cite{gkq22} and from
Kehoe's production example \cite{kehoe85num}. In each case only the selected price rays and
their local index labels are carried into the construction. The resulting monodromic family
is a new pure-exchange family; no monodromy claim is made about the source economy itself.

\section{Implications for comparative statics, selection, and computation}\label{sec:economic}

The economic implication of monodromy concerns the passage from local to global comparative
statics. Write the reduced equilibrium-price equations abstractly as
\[
F(x,\lambda)=0.
\]
At a regular normalized equilibrium price $(x_0,\lambda_0)$, invertibility of
$D_xF(x_0,\lambda_0)$ gives a unique local equilibrium-price section. This is exactly the
standard local conclusion of the implicit function theorem, and our result does not
challenge it. What fails is the additional inference that these local sections necessarily
define a globally consistent identity for the equilibrium branch.

An equilibrium-price continuation along a parameter path is path lifting in the restricted
regular covering. Starting from a mobile price, every sufficiently short path segment has a
unique lift, yet concatenating those local lifts around the monodromic loop returns to the
other mobile price. Thus
\[
\text{unique local price continuation}
\quad\not\Rightarrow\quad
\text{path-independent global identification}.
\]
Endpoint fundamentals alone therefore need not identify the equilibrium reached by
continuation. Starting from the same benchmark equilibrium, one can remain at the benchmark
economy or move around a closed path and return to that very same economy; the two
continuation procedures can end at different equilibrium prices.

\subsection{Allocative content and equilibrium selection}

The transposition of the two mobile equilibrium prices is not merely a relabelling of the
same economic outcome. At the base parameter of Remark~\ref{rem:prices}, consumer $2$ is a
fixed symmetric Cobb--Douglas consumer with endowment $k\mathbf e_2$ in both exact
realizations. Since $p_2=1$, its wealth is $k$, and its demand for the first good is
\[
x^2_1(q_\pm)=\frac{k}{L\,q_{\pm,1}}
=\frac{k}{L\bigl(1\pm\sqrt{S(\lambda_0)}\bigr)}.
\]
Hence $x^2_1(q_+)\ne x^2_1(q_-)$. The two mobile equilibrium prices therefore support
different allocations even for a consumer whose preferences and endowment remain fixed
throughout the loop. The monodromy has allocative content, although no welfare comparison
between the two equilibria is implied.

The selection implication is more specific. Monodromy does not imply that the equilibrium
correspondence admits no global continuous selection. In our construction the fixed branch
$q_0$ is itself such a selection. What fails is the global continuous identification of
either member of the mobile pair. Starting from one mobile equilibrium and continuing it
once around the loop leads to the other; two circuits are required to return to the initial
branch.

An additional equilibrium-selection rule can therefore resolve an ambiguity that
continuation alone does not resolve. Loi, Matta, and Uccheddu \cite{lmu23}, for example,
propose a geometric procedure for selecting an equilibrium price after changes in
endowments. Such a procedure supplies an independent selection criterion. Our result is
narrower: it shows that local continuation by itself need not provide a path-independent
global identification of an equilibrium branch.

The two-good case illustrates why this problem does not arise in one normalized price
dimension. Distinct regular equilibria can be ordered globally by their relative price, and
this rank is preserved under continuation. With two or more normalized price dimensions,
this order-based labelling mechanism is no longer available.

\subsection{Continuation algorithms and path dependence}

Continuation methods compute an equilibrium by starting from a known benchmark equilibrium
and following a branch along a chosen path in parameter space \cite{eaves99}. When
equilibrium branches have nontrivial monodromy, the equilibrium obtained at the endpoint may
depend on that path, even when different paths lead to the same final economic primitives.

This matters whenever the reported counterfactual equilibrium is defined as the
continuation of a particular benchmark equilibrium. In that case, the initial equilibrium
and the endpoint fundamentals need not be sufficient to identify the computed outcome: the
continuation path may also matter. Unless an independent selection rule makes the endpoint
path independent, reproducibility therefore requires specifying the path used by the
algorithm:
\[
\boxed{\text{the continuation path can be part of the specification of the computed equilibrium.}}
\]

Methods designed to compute the entire equilibrium set face a different issue. Set-oriented
algebraic procedures, including semialgebraic and Gr\"obner-basis methods
\cite{kubler10jet,kubler10or}, do not define their output as the endpoint of a single
continued branch. They are therefore not subject to this particular form of path
dependence. Monodromy remains relevant, however, to the global organization and labelling
of the equilibria across parameter values.

This is a computational implication, not a behavioural one. Continuation describes how a
numerical procedure follows a solution branch; it is not a t\^atonnement process and does
not imply that an actual economy moves between equilibria in the same way. In particular,
our result does not establish dynamic irreversibility or hysteresis.

\section{Conclusion}\label{sec:conclusion}

We have shown that local regularity does not guarantee global identifiability of a
Walrasian equilibrium price. For every number of commodities $L\ge3$, there exists a closed
one-parameter family of pure-exchange economies for which all primitives return exactly to
their initial values while continuation around the loop exchanges two regular equilibrium
prices and leaves a third unchanged. With one or two goods this cannot happen. Three goods
are therefore the sharp threshold for nontrivial real monodromy in the class studied here.

The same aggregate experiment admits two exact economic realizations. In the first, all
endowments are fixed and only one consumer's preference varies. In the stronger
realization, every preference is fixed and only the distribution of one endowment among
three consumers changes; their total endowment, and hence aggregate resources, remains
constant. The substantive realization problem was therefore not whether each aggregate
field could be microfounded separately, but whether the entire loop could be generated
coherently by a continuous loop of economic primitives.

The transposition is structurally stable. Sufficiently small $C^1$ perturbations may move
the equilibrium prices, but they preserve the three relevant regular branches and their
monodromy. The construction is also not tied to the particular numerical fibre used in the
formal model: it can be constructed with prescribed positive price rays and the required
local index pattern. Our control is restricted to the compact target region, however; we do not
claim that the three displayed equilibria exhaust the complete global equilibrium set.

The economic implication is a limit on global comparative statics under multiplicity.
Regularity makes an equilibrium locally traceable, but local continuations need not define
a path-independent global identity. The exchanged prices support genuinely different
demands, so the phenomenon is not mere relabelling; yet it is a structural statement about
the equilibrium correspondence, not a t\^atonnement or a theorem of dynamic hysteresis.
The identity of a regular equilibrium can therefore be locally unambiguous and yet globally
path dependent.

\subsection*{Acknowledgements}
Stefano Matta gratefully acknowledges financial support from Fondazione di Sardegna, grant no. F83C26000350007. Andrea Loi was supported by ProBiki of Fondazione di Sardegna.

\subsection*{Data and code availability}
No empirical data are used. The verification scripts are provided as separate replication
material; the proofs in the paper are analytic. A permanent repository identifier can be
added at submission.

\appendix

\section{Analytic construction in dimension three}\label{app:analytic}

This appendix supplies the complete analytic proof of Theorem~\ref{thm:parametric}. The
notation, regions, standing constants, and explicit potentials are those of
Section~\ref{sec:realization}. The proof separates the parameter-dependent part from the
analytic machinery that is uniform in the parameter. Its steps are: extend the explicit
potentials without changing them on the target cone; add Mantel's logarithmic benchmark;
choose one value of $k$ that gives uniform convexity and positive normalized demand; prove
that the negative-gradient demand map is globally invertible; construct direct utilities
and their Marshallian demands; and aggregate the three explicit consumers. The resulting
block equals the target field plus a universal parameter-independent residual, which is
removed by one fixed application of Debreu's realization theorem. The only ingredients
specific to our formal field are therefore the explicit potentials and the fact that their
parameter dependence is concentrated in one of them.

The sign convention for the gauge cores is chosen to match the negative-gradient demand
construction used below.

\begin{definition}[The gauge cores]\label{def:phicore}
Set $\varphi^{i}_{\lambda}\deq -\,g_{i}$, i.e.
\[
\varphi^{1}(u,v)=2cv^{2},\qquad \varphi^{2}=0,\qquad
\varphi^{3}_{\lambda}(u,v)=-\,g_{3}(u,v;\lambda),
\]
so that, by Proposition~\ref{prop:potential},
\begin{equation}\label{eq:minusZ}
\sum_{i=1}^{3}p_{i}\,\nabla_{p}\,\widehat\varphi^{\,i}_{\lambda}(p)\;=\;-\,Z(p;\lambda)
\qquad(p\in\R^{3}_{++},\ \lambda\in\R^{2}),
\end{equation}
where $\widehat\varphi^{\,i}_{\lambda}(p)=\varphi^{i}_{\lambda}(u(p),v(p))$. Only
$\varphi^{3}$ depends on $\lambda$, and it does so polynomially through
$(\sigma_{1},\sigma_{2})$.
\end{definition}

\subsection{The homogeneous cutoff extension}\label{sec:cutoff}

The gauge cores $\varphi^{i}_{\lambda}$ are polynomials in $(u,v)$: unbounded on the orthant
and with no uniform Hessian control. We extend them globally by the cutoff of
Assumption~\ref{ass:data}, exploiting that $u,v$ are themselves hom-$0$, so that cutting off
in $(u,v)$ \emph{preserves homogeneity exactly}. We need this global control in order to construct utilities and demands on the whole positive
orthant, but the cutoff must be flat on a neighborhood of the target cone so that the
exact potential and gradient identities used for realization are not altered there.

\begin{definition}[Cutoff extension]\label{def:phihat}
For $p\in\R^{3}_{++}$, $\lambda\in\R^{2}$, $i\in\{1,2,3\}$, set
\[
\ph^{\,i}_{\lambda}(p)\;\deq\;\chi\bigl(u(p),v(p)\bigr)\,
\varphi^{i}_{\lambda}\bigl(u(p),v(p)\bigr).
\]
\end{definition}

\begin{lemma}[Properties of the extension]\label{lem:extension}
The functions $\ph^{\,i}_{\lambda}$ satisfy:
\begin{enumerate}[label=\textup{(\roman*)},leftmargin=2.2em]
\item $(p,\lambda)\mapsto\ph^{\,i}_{\lambda}(p)$ is $C^{\infty}$ on
$\R^{3}_{++}\times\R^{2}$, and $\ph^{\,i}_{\lambda}$ is hom-$0$ in $p$; $\ph^{\,1},\ph^{\,2}$
do not depend on $\lambda$;
\item \textbf{(Flatness on the target cone.)} $\ph^{\,i}_{\lambda}=\widehat\varphi^{\,i}_{\lambda}$
on an open set containing $\Gamma=\Gamma_{\rho_{1}}$; in particular
\emph{all $p$-derivatives} of $\ph^{\,i}_{\lambda}$ and $\widehat\varphi^{\,i}_{\lambda}$
coincide at every point of $\Gamma$, and \eqref{eq:minusZ} holds with
$\widehat\varphi$ replaced by $\ph$ for all $p\in\Gamma$;
\item there are finite constants
\[
\begin{aligned}
B&\deq\sup\Bigl\{\bigl|\ph^{\,i}_{\lambda}(q)\bigr|
:\ i\in\{1,2,3\},\ \lambda\in\Lambda_{1},\ q\in\R^{3}_{++}\Bigr\},\\
C&\deq\sup\Bigl\{\norm{\widehat q\,D^{2}_{q}\ph^{\,i}_{\lambda}(q)\,\widehat q}_{\mathrm{op}}
:\ i\in\{1,2,3\},\ \lambda\in\Lambda_{1},\ q\in\R^{3}_{++}\Bigr\},\\
C_{1}&\deq\sup\Bigl\{\bigl|q_{j}\,\partial_{j}\ph^{\,i}_{\lambda}(q)\bigr|
:\ i,j\in\{1,2,3\},\ \lambda\in\Lambda_{1},\ q\in\R^{3}_{++}\Bigr\}.
\end{aligned}
\]
\end{enumerate}
\end{lemma}

\begin{proof}
(i) Smoothness: $u,v$ are smooth on $\R^{3}_{++}$ (where $p_{3}>0$ always holds), $\chi$ and
$\varphi^{i}_{\lambda}$ are smooth, and $\varphi^{i}_{\lambda}$ is polynomial in $\lambda$.
Homogeneity: $u(tp)=u(p)$, $v(tp)=v(p)$.

(ii) By (D4), $\chi\equiv1$ on an open set $U\supset\overline D_{\rho_{1}}$; hence
$\ph^{\,i}_{\lambda}$ and $\widehat\varphi^{\,i}_{\lambda}$ coincide on the open cone
$\{p:(u,v)\in U\}\supset\Gamma$, so all derivatives coincide on $\Gamma$. This is the point of hypothesis (D4): coincidence of the \emph{fields} on $K$ alone would not
transfer derivative identities; coincidence on an open neighborhood does.

(iii) Each displayed quantity is hom-$0$ in $q$: $\abs{\ph}$ trivially;
$q_{j}\partial_{j}\ph$ because $\nabla\ph$ is hom-$(-1)$ by \eqref{eq:homgrad}; and
$\widehat q\,D^{2}\ph\,\widehat q$ because $D^{2}\ph$ is hom-$(-2)$ while each factor
$\widehat q$ is linear in $q$:
$\widehat{tq}\,D^{2}\ph(tq)\,\widehat{tq}=t\widehat q\cdot t^{-2}D^{2}\ph(q)\cdot t\widehat q
=\widehat q\,D^{2}\ph(q)\,\widehat q$. A hom-$0$ function is determined by its values on the
slice $\{q_{3}=1\}$, where all three quantities vanish outside
$\{(u,v)\in\supp\chi\subset\overline D_{\rho_{2}}\}$; on that support,
$q_{1}=1+u$ and $q_{2}=1+v$ range in $[1-\rho_{2},\,1+\rho_{2}]$ and $q_{3}=1$, a compact subset
$\mathcal{Q}\subset\R^{3}_{++}$. The three quantities are continuous on the compact
$\mathcal{Q}\times\Lambda_{1}$ (joint smoothness of $(q,\lambda)\mapsto\ph^{\,i}_{\lambda}$),
hence bounded; taking the maximum over $i\in\{1,2,3\}$ gives finite $B$, $C$, $C_{1}$.
\end{proof}

The three bounds have distinct roles below: $C$ controls the Hessian and hence strict
convexity, $C_{1}$ guarantees coordinatewise monotonicity and positive demand, and $B$
is used in the global inverse argument to control the minimizing problem near the
boundary of the price orthant.

\subsection{The gauges and the uniform choice of \texorpdfstring{$k$}{k}}\label{sec:gauges}

We now add Mantel's logarithmic benchmark. The perturbation $\ph^{\,i}_{\lambda}/k$
carries the prescribed aggregate field, while the logarithmic term supplies a globally
well-behaved strictly convex benchmark. Choosing $k$ sufficiently large makes this
benchmark dominate the first and second derivatives of the perturbation without changing
the exact weighted-gradient information needed on the target cone.

\begin{definition}[Gauges; cf.\ \cite{mantel74}]\label{def:gauge}
For $k>0$, $\lambda\in\R^{2}$, $i\in\{1,2,3\}$, define
$\psi^{i}_{\lambda}\colon\R^{3}_{++}\to\R$,
\[
\psi^{i}_{\lambda}(q)\;\deq\;\frac{\ph^{\,i}_{\lambda}(q)}{k}
\;-\;\frac13\sum_{j=1}^{3}\log q_{j}.
\]
\end{definition}

\begin{remark}[Parameter domains]
The gauges and their underlying potentials are defined for every $\lambda\in\R^{2}$, since
all data are polynomial in $\lambda$. The uniform estimates are asserted on
$\Lambda_{1}$; the derived objects ($Q^{i}$, $u^{i}_{\lambda}$, $U^{i}_{\lambda}$ and the
demands) are constructed there, with smooth parameter dependence on the open set
$W=\operatorname{int}\Lambda_{1}$. The economic family ultimately used in the main theorem
is restricted to $\Lambda\subset W$.
\end{remark}

\begin{lemma}[Exact identities: log-homogeneity and radial Euler]\label{lem:identities}
For every $k>0$, $\lambda\in\R^{2}$, $i$, and all $q\in\R^{3}_{++}$, $t>0$:
\begin{equation}\label{eq:loghom}
\psi^{i}_{\lambda}(tq)=\psi^{i}_{\lambda}(q)-\log t,
\end{equation}
\begin{equation}\label{eq:radial}
q\cdot\nabla\psi^{i}_{\lambda}(q)=-1,
\end{equation}
\begin{equation}\label{eq:diffeuler}
D^{2}\psi^{i}_{\lambda}(q)\,q=-\nabla\psi^{i}_{\lambda}(q),
\end{equation}
\begin{equation}\label{eq:gradhom}
\nabla\psi^{i}_{\lambda}(tq)=t^{-1}\nabla\psi^{i}_{\lambda}(q).
\end{equation}
None of these requires any condition on $k$.
\end{lemma}

\begin{proof}
\eqref{eq:loghom}: $\ph$ is hom-$0$ and $\sum_{j}\log(tq_{j})=\sum_{j}\log q_{j}+3\log t$.
\eqref{eq:radial}: by Euler's identity \eqref{eq:euler0}, $q\cdot\nabla\ph=0$, and
$q\cdot\nabla\bigl(-\tfrac13\sum\log q_{j}\bigr)=-\tfrac13\sum_{j}q_{j}\cdot q_{j}^{-1}=-1$.
This is an \emph{exact identity}; the largeness of $k$ plays no role here (it will be needed
only for convexity, coordinatewise monotonicity, and demand positivity).
\eqref{eq:diffeuler}: differentiate the identity \eqref{eq:radial} in $q$:
$\nabla\psi+D^{2}\psi\,q=0$.
\eqref{eq:gradhom}: differentiate \eqref{eq:loghom} in $q$.
\end{proof}

\begin{lemma}[Uniform strict convexity]\label{lem:convexity}
If $k>3C$ then, for all $\lambda\in\Lambda_{1}$, $i\in\{1,2,3\}$, $q\in\R^{3}_{++}$, and
$\xi\in\R^{3}$,
\[
\xi^{\top}D^{2}\psi^{i}_{\lambda}(q)\,\xi\;\ge\;
\Bigl(\frac13-\frac{C}{k}\Bigr)\,\bigl|\widehat q^{-1}\xi\bigr|^{2},
\qquad\text{i.e.}\qquad
D^{2}\psi^{i}_{\lambda}(q)\succeq\Bigl(\frac13-\frac{C}{k}\Bigr)\widehat q^{-2}\succ0 ,
\]
with $C$ from Lemma~\ref{lem:extension}(iii). In particular $\psi^{i}_{\lambda}$ is strictly
convex on $\R^{3}_{++}$, uniformly over $\Lambda_{1}$, with the same $k$ for all $(i,\lambda)$.
\end{lemma}

\begin{proof}
$D^{2}\psi=\tfrac1kD^{2}\ph+\tfrac13\diag(q_{j}^{-2})=\tfrac1kD^{2}\ph+\tfrac13\widehat q^{-2}$.
Substituting $\eta\deq\widehat q^{-1}\xi$,
\[
\bigl|\xi^{\top}D^{2}\ph\,\xi\bigr|
=\bigl|\eta^{\top}\bigl(\widehat q\,D^{2}\ph\,\widehat q\bigr)\eta\bigr|
\le C\,\abs{\eta}^{2},
\qquad
\xi^{\top}\widehat q^{-2}\xi=\abs{\eta}^{2},
\]
so $\xi^{\top}D^{2}\psi\,\xi\ge(\tfrac13-\tfrac{C}{k})\abs{\eta}^{2}$. Since
$C$ is a single constant valid for all $q\in\R^{3}_{++}$ (scale invariance,
Lemma~\ref{lem:extension}(iii)), the bound holds on the whole orthant.
\end{proof}

\begin{lemma}[Coordinatewise monotonicity and positivity of the normalized demand]
\label{lem:positivity}
If $k>3C_{1}$ then, for all $\lambda\in\Lambda_{1}$, $i$, $q\in\R^{3}_{++}$, and $j\in\{1,2,3\}$,
\[
-\,q_{j}\,\partial_{j}\psi^{i}_{\lambda}(q)\;\ge\;\frac13-\frac{C_{1}}{k}\;>\;0.
\]
Consequently $\partial_{j}\psi^{i}_{\lambda}<0$ for each $j$ (strict coordinatewise
decrease) and
\[
G^{i}_{\lambda}(q)\;\deq\;-\nabla\psi^{i}_{\lambda}(q)\;\in\;\R^{3}_{++},
\qquad
\bigl(G^{i}_{\lambda}(q)\bigr)_{j}\ \ge\ \Bigl(\frac13-\frac{C_{1}}{k}\Bigr)\frac{1}{q_{j}}.
\]
\end{lemma}

\begin{proof}
$q_{j}\partial_{j}\psi=\tfrac1k\,q_{j}\partial_{j}\ph-\tfrac13$ and
$\abs{q_{j}\partial_{j}\ph}\le C_{1}$ by Lemma~\ref{lem:extension}(iii) (absolute value, as
required for a two-sided uniform bound). Hence
$q_{j}\partial_{j}\psi\le\tfrac{C_{1}}{k}-\tfrac13<0$ and
$-q_{j}\partial_{j}\psi\ge\tfrac13-\tfrac{C_{1}}{k}>0$.
\end{proof}

One value of $k$, chosen independently of both consumer and parameter, can therefore enforce all the properties needed below.

\begin{definition}[Standing choice of $k$]\label{def:k}
Fix, once and for all,
\[
k\;\deq\;1+3\max(C,\,C_{1}),
\]
so that the conclusions of Lemmas~\ref{lem:convexity} and~\ref{lem:positivity} hold
simultaneously, uniformly over $\Lambda_{1}$ and $i\in\{1,2,3\}$. We henceforth write
$m_{k}\deq\tfrac13-\tfrac{\max(C,C_{1})}{k}>0$ for the common margin.
\end{definition}

\subsection{The inverse-demand map is a global diffeomorphism}\label{sec:diffeo}

The map $G^{i}_{\lambda}=-\nabla\psi^{i}_{\lambda}$ is the normalized demand candidate.
To recover a direct utility from it, every strictly positive bundle must correspond to one
and only one supporting normalized price vector. The next proposition establishes this
global invertibility. Its budget identity says that $G^{i}_{\lambda}(q)$ automatically lies
on the unit-income budget hyperplane at prices $q$.

\begin{proposition}[Global diffeomorphism]\label{prop:diffeo}
With $k$ as in Definition~\ref{def:k}, for every $\lambda\in\Lambda_{1}$ and $i\in\{1,2,3\}$ the
map
\[
G^{i}_{\lambda}\;=\;-\nabla\psi^{i}_{\lambda}\;\colon\;\R^{3}_{++}\longrightarrow\R^{3}_{++}
\]
is a $C^{\infty}$ diffeomorphism of $\R^{3}_{++}$ onto $\R^{3}_{++}$. Moreover every
$q\in\R^{3}_{++}$ satisfies the \emph{automatic budget identity}
\begin{equation}\label{eq:autobudget}
q\cdot G^{i}_{\lambda}(q)=1.
\end{equation}
Denoting the inverse by $Q^{i}(\,\cdot\,,\lambda)\deq\bigl(G^{i}_{\lambda}\bigr)^{-1}$, the
map $(x,\lambda)\mapsto Q^{i}(x,\lambda)$ is $C^{\infty}$ on
$\R^{3}_{++}\times W$ (jointly), with
\begin{equation}\label{eq:DQ}
D_{x}Q^{i}(x,\lambda)=-\bigl(D^{2}\psi^{i}_{\lambda}(Q^{i})\bigr)^{-1},
\qquad
D_{\lambda}Q^{i}(x,\lambda)
=-\bigl(D^{2}\psi^{i}_{\lambda}(Q^{i})\bigr)^{-1}D_{\lambda}\nabla_{q}\psi^{i}_{\lambda}(Q^{i});
\end{equation}
\end{proposition}

\begin{proof}
Fix $\lambda,i$ and drop them from the notation.

\emph{Values in $\R^{3}_{++}$ and \eqref{eq:autobudget}.} By
Lemma~\ref{lem:positivity}, $G(q)=-\nabla\psi(q)\gg0$. Identity \eqref{eq:autobudget} is the
radial Euler identity \eqref{eq:radial}: $q\cdot G(q)=-q\cdot\nabla\psi(q)=1$.

\emph{Injectivity.} $\R^{3}_{++}$ is convex; for $q\ne q'$, integrating the Hessian along
the segment $q'+t(q-q')\subset\R^{3}_{++}$ and using Lemma~\ref{lem:convexity},
\[
\bigl(\nabla\psi(q)-\nabla\psi(q')\bigr)\cdot(q-q')
=\int_{0}^{1}(q-q')^{\top}D^{2}\psi\bigl(q'+t(q-q')\bigr)(q-q')\,dt\;>\;0,
\]
so $\nabla\psi$, hence $G$, is injective.

\emph{Surjectivity.} Let $x\in\R^{3}_{++}$ and consider
\begin{equation}\label{eq:minproblem}
\min\bigl\{\psi(q)\ :\ q\in\R^{3}_{++},\ q\cdot x\le1\bigr\}.
\end{equation}
The feasible set is nonempty ($q=\varepsilon\one$ for small $\varepsilon>0$) and bounded
($0<q_{j}\le1/x_{j}$). On it, $\abs{\ph(q)/k}\le B/k$ by Lemma~\ref{lem:extension}(iii),
while $-\tfrac13\sum_{j}\log q_{j}\ge-\tfrac13\log q_{j_{0}}+\tfrac13\sum_{j\ne j_{0}}\log x_{j}
\to+\infty$ as any $q_{j_{0}}\to0$ within the feasible set. Hence any minimizing sequence
stays in a compact subset of the open orthant and, by continuity, the minimum is attained at
some $q^{*}\gg0$; it is unique because $\psi$ is strictly convex
(Lemma~\ref{lem:convexity}) on the convex feasible set. The minimizer lies on the face
$q\cdot x=1$: if $q^{*}\cdot x<1$, then $t\deq1/(q^{*}\cdot x)>1$ gives a feasible point
$tq^{*}$ with, by \eqref{eq:loghom},
$\psi(tq^{*})=\psi(q^{*})-\log t<\psi(q^{*})$, a contradiction. At an interior-of-orthant
minimizer on the smooth constraint surface $\{q\cdot x=1\}$ (whose gradient $x$ never
vanishes), the Lagrange condition holds: there is $\mu\in\R$ with
\[
\nabla\psi(q^{*})+\mu x=0.
\]
Taking the inner product with $q^{*}$ and using \eqref{eq:radial} and $q^{*}\cdot x=1$:
$-1+\mu=0$, so $\mu=1$ and $x=-\nabla\psi(q^{*})=G(q^{*})$. Thus $G$ is onto $\R^{3}_{++}$.

\emph{Smoothness of the inverse, jointly in $(x,\lambda)$.} Define the $C^{\infty}$ map
\[
H\colon\R^{3}_{++}\times\R^{3}_{++}\times W\to\R^{3},
\qquad
H(q,x,\lambda)\deq\nabla_{q}\psi^{i}_{\lambda}(q)+x ,
\]
so that $G^{i}_{\lambda}(q)=x\iff H(q,x,\lambda)=0$. At any zero,
$D_{q}H=D^{2}\psi^{i}_{\lambda}(q)\succeq m_{k}\,\widehat q^{-2}\succ0$ is invertible
(Lemma~\ref{lem:convexity}), so the implicit function theorem yields a local $C^{\infty}$
solution $q=Q^{i}(x,\lambda)$; by the already-proved global bijectivity these local
solutions agree and patch into a single $C^{\infty}$ map on
$\R^{3}_{++}\times W$. Differentiating the identity $H(Q^{i},x,\lambda)=0$ in
$x$ gives $D^{2}\psi\cdot D_{x}Q^{i}+I=0$, and in $\lambda$ gives
$D^{2}\psi\cdot D_{\lambda}Q^{i}+D_{\lambda}\nabla_{q}\psi=0$; solving yields
\eqref{eq:DQ}. This is the explicit map to which the implicit function theorem is applied;
no bordered system is needed, because the budget constraint is not an extra
equation: by \eqref{eq:autobudget} it holds \emph{identically} along the relation
$x=G(q)$.
\end{proof}

\subsection{The direct utilities}\label{sec:direct}

\subsubsection{Interior utility and demand}

Having inverted the normalized demand map, we can associate to every strictly positive
bundle $x$ the unique normalized price $Q^{i}(x,\lambda)$ that supports it. Evaluating the
gauge at that supporting price produces the direct utility representation. The equivalent
minimization formulas below make explicit the connection with Mantel's construction.

\begin{definition}[Direct utility]\label{def:direct}
For $i\in\{1,2,3\}$, $\lambda\in\Lambda_{1}$, define $u^{i}_{\lambda}\colon\R^{3}_{++}\to\R$ by
\begin{equation}\label{eq:udef}
u^{i}_{\lambda}(x)\;\deq\;\psi^{i}_{\lambda}\bigl(Q^{i}(x,\lambda)\bigr).
\end{equation}
(We write $u^{1},u^{2}$ without the subscript $\lambda$, since $\psi^{1},\psi^{2}$ do not
depend on $\lambda$.)
\end{definition}

\begin{lemma}[Three equivalent formulas]\label{lem:equivalence}
For every $x\in\R^{3}_{++}$,
\[
\begin{aligned}
u^{i}_{\lambda}(x)
&=\min\bigl\{\psi^{i}_{\lambda}(q):q\in\R^{3}_{++},\ q\cdot x\le1\bigr\}
=\min\bigl\{\psi^{i}_{\lambda}(q):q\in\R^{3}_{++},\ q\cdot x=1\bigr\}\\
&=\inf_{q\in\R^{3}_{++}}\psi^{i}_{\lambda}\!\Bigl(\frac{q}{q\cdot x}\Bigr),
\end{aligned}
\]
and the (unique) minimizer of the first two problems is $Q^{i}(x,\lambda)$. The third
formula is Mantel's \cite[Remark~1]{mantel74}.
\end{lemma}

\begin{proof}
In the proof of Proposition~\ref{prop:diffeo} it was shown that problem
\eqref{eq:minproblem} has a unique minimizer $q^{*}$, that it lies on the face
$q\cdot x=1$ (by the log-homogeneity \eqref{eq:loghom}: any feasible $q$ with $q\cdot x<1$
is strictly improved by the feasible rescaling $q/(q\cdot x)$), and that it satisfies
$x=G(q^{*})$, i.e.\ $q^{*}=Q^{i}(x,\lambda)$. This proves the first two equalities and the
identification of the minimizer. For the third: as $q$ ranges over $\R^{3}_{++}$, the point
$q/(q\cdot x)$ ranges exactly over the face $\{q'\gg0:q'\cdot x=1\}$, so the infimum of
$\psi(q/(q\cdot x))$ equals the minimum over that face.
\end{proof}

\begin{theorem}[Regularity, monotonicity, concavity]\label{thm:direct}
With $k$ as in Definition~\ref{def:k}:
\begin{enumerate}[label=\textup{(\alph*)},leftmargin=2.2em]
\item $(x,\lambda)\mapsto u^{i}_{\lambda}(x)$ is $C^{\infty}$ on
$\R^{3}_{++}\times W$; $u^{1},u^{2}$ are independent of $\lambda$;
\item $\nabla_{x}u^{i}_{\lambda}(x)=Q^{i}(x,\lambda)\gg0$; in particular $u^{i}_{\lambda}$
is \emph{strongly monotone}: $x'\ge x$, $x'\ne x$ imply $u^{i}_{\lambda}(x')>u^{i}_{\lambda}(x)$;
\item $D^{2}_{x}u^{i}_{\lambda}(x)=-\bigl(D^{2}\psi^{i}_{\lambda}(Q^{i}(x,\lambda))\bigr)^{-1}\prec0$
for every $x$; since $\R^{3}_{++}$ is convex, $u^{i}_{\lambda}$ is \emph{strictly concave}
on its whole domain, hence strictly quasiconcave.
\end{enumerate}
\end{theorem}

\begin{proof}
(a) Composition of the $C^{\infty}$ maps $(x,\lambda)\mapsto(Q^{i}(x,\lambda),\lambda)$
(Proposition~\ref{prop:diffeo}) and $(q,\lambda)\mapsto\psi^{i}_{\lambda}(q)$.

(b) By Lemma~\ref{lem:equivalence}, $u(x)=\psi(q)+\mu\,(q\cdot x-1)$ at
$(q,\mu)=(Q(x,\lambda),1)$ is a constrained value function with multiplier $\mu=1$
(established in Proposition~\ref{prop:diffeo}); we verify the envelope conclusion by direct
computation rather than by citation. Chain rule on \eqref{eq:udef}:
\[
\nabla_{x}u
=\bigl(D_{x}Q\bigr)^{\top}\nabla_{q}\psi(Q)
=\Bigl(-\bigl(D^{2}\psi(Q)\bigr)^{-1}\Bigr)^{\top}\bigl(-x\bigr)
=\bigl(D^{2}\psi(Q)\bigr)^{-1}x ,
\]
using \eqref{eq:DQ}, the symmetry of the Hessian, and $\nabla_{q}\psi(Q)=-x$ (definition of
$Q$). By the differentiated Euler identity \eqref{eq:diffeuler}, $\bigl(D^{2}\psi(Q)\bigr)^{-1}x=Q(x,\lambda)$. Hence
$\nabla_{x}u=Q\gg0$. Strong monotonicity: for $x'\ge x$, $x'\ne x$, the segment
$x+t(x'-x)$ lies in $\R^{3}_{++}$ and
\[
u(x')-u(x)=\int_{0}^{1}\nabla_{x}u\bigl(x+t(x'-x)\bigr)\cdot(x'-x)\,dt>0,
\]
because the integrand is a sum of nonnegative terms with at least one strictly positive
($\nabla_{x}u\gg0$ and $x'-x\ge0$, $\ne0$).

(c) Differentiating $\nabla_{x}u(x)=Q(x,\lambda)$ in $x$ and using \eqref{eq:DQ},
\[
D^{2}_{x}u(x)=D_{x}Q=-\bigl(D^{2}\psi(Q)\bigr)^{-1}\prec0 ,
\]
since $D^{2}\psi\succ0$ (Lemma~\ref{lem:convexity}). A $C^{2}$ function with everywhere
negative-definite Hessian on a convex open set is strictly concave; strict concavity implies
strict quasiconcavity.
\end{proof}

\begin{theorem}[Exact demand rationalization]\label{thm:rationalize}
Let $k$ be as in Definition~\ref{def:k}, $i\in\{1,2,3\}$, $\lambda\in\Lambda_{1}$. For every
$(p,m)\in\R^{3}_{++}\times(0,\infty)$, the problem
\[
\max\bigl\{u^{i}_{\lambda}(x)\ :\ x\in\R^{3}_{++},\ p\cdot x\le m\bigr\}
\]
has the unique solution
\begin{equation}\label{eq:demand}
x^{i}_{\lambda}(p,m)\;=\;-\nabla\psi^{i}_{\lambda}\!\Bigl(\frac{p}{m}\Bigr)
\;=\;-\,m\,\nabla\psi^{i}_{\lambda}(p)
\;=\;\frac{m}{3}\,\widehat p^{\,-1}\one\;-\;\frac{m}{k}\,\nabla\ph^{\,i}_{\lambda}(p),
\end{equation}
which satisfies the budget constraint \emph{with equality}, $p\cdot x^{i}_{\lambda}(p,m)=m$,
is strictly positive in every component, is homogeneous of degree zero in $(p,m)$, and is
$C^{\infty}$ jointly in $(p,m,\lambda)$.
\end{theorem}

\begin{proof}
Set $q\deq p/m\in\R^{3}_{++}$ and $\bar x\deq G^{i}_{\lambda}(q)=-\nabla\psi^{i}_{\lambda}(q)$.

\emph{Formulas.} The second equality in \eqref{eq:demand} is the degree-$(-1)$ homogeneity
of $\nabla\psi$ \eqref{eq:gradhom}; the third is
$\nabla\psi(p)=\tfrac1k\nabla\ph(p)-\tfrac13\widehat p^{\,-1}\one$. This is Mantel's demand.

\emph{Budget with equality and positivity.} By the automatic budget identity
\eqref{eq:autobudget}, $q\cdot\bar x=1$, i.e.\ $p\cdot\bar x=m$: the constraint binds
exactly, as an identity. Positivity is Lemma~\ref{lem:positivity}, quantitatively
$\bar x_{j}\ge m\,m_{k}/p_{j}>0$.

\emph{Optimality.} Let $x$ be feasible: $p\cdot x\le m$, i.e.\ $q\cdot x\le1$. Then $q$ is
feasible for the minimization defining $u^{i}_{\lambda}(x)$ (Lemma~\ref{lem:equivalence}),
so
\[
u^{i}_{\lambda}(x)\le\psi^{i}_{\lambda}(q).
\]
On the other hand, since $Q^{i}(\bar x,\lambda)=\bigl(G^{i}_{\lambda}\bigr)^{-1}\bigl(G^{i}_{\lambda}(q)\bigr)=q$,
Definition~\ref{def:direct} gives $u^{i}_{\lambda}(\bar x)=\psi^{i}_{\lambda}(q)$. Hence
$u^{i}_{\lambda}(x)\le u^{i}_{\lambda}(\bar x)$ for every feasible $x$: $\bar x$ is a
maximizer.

\emph{Uniqueness.} If $x\ne\bar x$ were another maximizer, then by strict concavity
(Theorem~\ref{thm:direct}(c)) the feasible point $\tfrac12(x+\bar x)$ (the budget set is
convex) would satisfy $u(\tfrac12(x+\bar x))>\tfrac12u(x)+\tfrac12u(\bar x)=u(\bar x)$,
contradicting maximality.

\emph{Regularity and homogeneity.} $(p,m,\lambda)\mapsto-\nabla\psi^{i}_{\lambda}(p/m)$ is a
composition of $C^{\infty}$ maps; and $(tp)/(tm)=p/m$ gives degree-zero homogeneity.
\end{proof}

\subsubsection{The extension to the closed orthant}\label{subsec:extension}

We now extend the block-1 utilities to the closed orthant, so that the consumers fit the
standard class of Definition~\ref{def:Meconomy}, and we show that the extension changes
nothing in the interior analysis. The comparison with the symmetric Cobb--Douglas
benchmark is used only to control boundary behavior: it bounds the constructed homogeneous
utility above and below by fixed multiples of the geometric mean. The engine is the
following uniform two-sided estimate.

\begin{proposition}[Uniform Cobb--Douglas sandwich]\label{prop:boundary}
Let $k$ be as in Definition~\ref{def:k}, $m_{k}=\tfrac13-\tfrac{\max(C,C_{1})}{k}>0$, and
set
\[
C^{*}\;\deq\;\frac{B}{k}+\log\frac{1}{m_{k}}\;<\;\infty.
\]
For every $i\in\{1,2,3\}$, $\lambda\in\Lambda_{1}$, and $x\in\R^{3}_{++}$:
\begin{enumerate}[label=\textup{(\roman*)},leftmargin=2.2em]
\item $\bigl|\,u^{i}_{\lambda}(x)-\tfrac13\sum_{j=1}^{3}\log x_{j}\,\bigr|\le C^{*}$;
equivalently, in multiplicative form,
\begin{equation}\label{eq:sandwich}
e^{-C^{*}}\,(x_{1}x_{2}x_{3})^{1/3}
\;\le\; e^{\,u^{i}_{\lambda}(x)}
\;\le\; e^{C^{*}}\,(x_{1}x_{2}x_{3})^{1/3};
\end{equation}
\item consequently $u^{i}_{\lambda}(x)\to-\infty$, uniformly in $(i,\lambda)$, whenever $x$
converges in $\R^{3}$ to a point of $\partial\R^{3}_{+}$; every upper contour set
$\{x\in\R^{3}_{++}:u^{i}_{\lambda}(x)\ge c\}$ is therefore closed in $\R^{3}$, i.e.\
indifference surfaces do not accumulate on the boundary of the orthant.
\end{enumerate}
\end{proposition}

\begin{proof}
(i) Let $q=Q^{i}(x,\lambda)$, so that $x=G^{i}_{\lambda}(q)$ and $q\cdot x=1$
\eqref{eq:autobudget}. From the componentwise bound of Lemma~\ref{lem:positivity},
$x_{j}=\bigl(G^{i}_{\lambda}(q)\bigr)_{j}\ge m_{k}/q_{j}$, i.e.\ $q_{j}x_{j}\ge m_{k}$;
from $q\cdot x=1$ with positive terms, $q_{j}x_{j}\le1$. Hence
\[
\log q_{j}=\log(q_{j}x_{j})-\log x_{j}\in\bigl[\log m_{k}-\log x_{j},\ -\log x_{j}\bigr],
\]
so that $-\tfrac13\sum_{j}\log q_{j}-\tfrac13\sum_{j}\log x_{j}\in[0,\log(1/m_{k})]$. Since
$u^{i}_{\lambda}(x)=\psi^{i}_{\lambda}(q)=\ph^{\,i}_{\lambda}(q)/k-\tfrac13\sum_{j}\log q_{j}$
and $\abs{\ph^{\,i}_{\lambda}}\le B$ (Lemma~\ref{lem:extension}(iii)), the claim follows;
exponentiating gives \eqref{eq:sandwich}.

(ii) If $x_{n}\to x^{*}\in\partial\R^{3}_{+}$ with $x^{*}$ finite, then some coordinate
tends to $0$ while all remain bounded, so $(x_{n,1}x_{n,2}x_{n,3})^{1/3}\to0$ and, by
\eqref{eq:sandwich}, $e^{u^{i}_{\lambda}(x_{n})}\to0$ uniformly in $(i,\lambda)$, i.e.\
$u^{i}_{\lambda}(x_{n})\to-\infty$. Hence no sequence in an upper contour set
$\{u\ge c\}$ accumulates at a finite boundary point; by continuity of $u^{i}_{\lambda}$ on
the open orthant, $\{u\ge c\}$ is closed in $\R^{3}$.
\end{proof}

\begin{proposition}[The extended utilities]\label{prop:extension}
For $i\in\{1,2,3\}$ and $\lambda\in\Lambda_{1}$ define
$U^{i}_{\lambda}\colon\R^{3}_{+}\to\R$,
\[
U^{i}_{\lambda}(x)\;\deq\;
\begin{cases}
\exp\bigl(u^{i}_{\lambda}(x)\bigr), & x\in\R^{3}_{++},\\[1mm]
0, & x\in\partial\R^{3}_{+}.
\end{cases}
\]
Then:
\begin{enumerate}[label=\textup{(\roman*)},leftmargin=2.2em]
\item \textbf{(Joint continuity on the closed orthant.)} $U^{i}_{\lambda}$ is continuous on
$\R^{3}_{+}$, and $(x,\lambda)\mapsto U^{i}_{\lambda}(x)$ is jointly continuous on
$\R^{3}_{+}\times\Lambda_{1}$: if $(x_{n},\lambda_{n})\to(x,\lambda)$ with
$x\in\partial\R^{3}_{+}$, then $U^{i}_{\lambda_{n}}(x_{n})\to0=U^{i}_{\lambda}(x)$, because
the constant $C^{*}$ in \eqref{eq:sandwich} does not depend on $\lambda\in\Lambda_{1}$;
\item \textbf{(Global monotonicity.)} $x'\ge x$ implies $U^{i}_{\lambda}(x')\ge
U^{i}_{\lambda}(x)$, with strict inequality whenever moreover $x'\ne x$ and
$x'\in\R^{3}_{++}$;
\item \textbf{(Global quasiconcavity.)} for all $x,y\in\R^{3}_{+}$ and $t\in[0,1]$,
$U^{i}_{\lambda}\bigl(tx+(1-t)y\bigr)\ge\min\bigl\{U^{i}_{\lambda}(x),U^{i}_{\lambda}(y)\bigr\}$;
hence the induced preference is convex on $\R^{3}_{+}$;
\item \textbf{(Interior strict properties.)} on $\R^{3}_{++}$, $U^{i}_{\lambda}$ is
$C^{\infty}$, strongly monotone and strictly quasiconcave, and it represents the same
preference as $u^{i}_{\lambda}$;
\item \textbf{(Local nonsatiation.)} every relative neighborhood in $\R^{3}_{+}$ of every
point of $\R^{3}_{+}$ contains a strictly preferred bundle.
\end{enumerate}
\end{proposition}

\begin{proof}
(i) Continuity at interior points is clear ($u^{i}_{\lambda}$ smooth). At a boundary point
$x$, the upper bound in \eqref{eq:sandwich} gives, for any sequence
$(x_{n},\lambda_{n})\to(x,\lambda)$ in $\R^{3}_{++}\times\Lambda_{1}$,
\[
0\le U^{i}_{\lambda_{n}}(x_{n})\le e^{C^{*}}\,(x_{n,1}x_{n,2}x_{n,3})^{1/3}
\longrightarrow e^{C^{*}}\,(x_{1}x_{2}x_{3})^{1/3}=0 ,
\]
since $C^{*}$ is uniform over $\Lambda_{1}$ and the geometric mean is continuous on
$\R^{3}_{+}$ and vanishes on the boundary; sequences approaching $x$ within the boundary
are trivial. Joint continuity follows.

(ii) If $x\in\R^{3}_{++}$, then $x'\ge x$ forces $x'\in\R^{3}_{++}$ and the claim (with
strictness for $x'\ne x$) is the strong monotonicity of $u^{i}_{\lambda}$
(Theorem~\ref{thm:direct}(b)) composed with $\exp$. If $x\in\partial\R^{3}_{+}$, then
$U^{i}_{\lambda}(x)=0\le U^{i}_{\lambda}(x')$, strictly when $x'\in\R^{3}_{++}$.

(iii) If $x,y\in\R^{3}_{++}$, the segment lies in $\R^{3}_{++}$ and the claim follows from
concavity of $u^{i}_{\lambda}$ and monotonicity of $\exp$. If at least one of $x,y$ lies on
the boundary, then $\min\{U(x),U(y)\}=0$ while $U\ge0$ everywhere, so the inequality holds
(in particular it holds, with equality $0=0$, when $x$, $y$ and the segment all lie in one
boundary face). Convexity of upper contour sets follows: $\{U\ge c\}$ equals $\R^{3}_{+}$
for $c\le0$ and $\{x\gg0:u^{i}_{\lambda}(x)\ge\log c\}$ for $c>0$, which is convex by
concavity of $u^{i}_{\lambda}$.

(iv) On the open orthant $\exp$ is a strictly increasing $C^{\infty}$ reparametrization, so
$U^{i}_{\lambda}$ and $u^{i}_{\lambda}$ have the same indifference relation, the same
strict monotonicity (Theorem~\ref{thm:direct}(b)) and the same strict quasiconcavity
(Theorem~\ref{thm:direct}(c)).

(v) At an interior point, strong monotonicity provides strictly preferred bundles in every
neighborhood. At a boundary point $x$, every relative neighborhood of $x$ in $\R^{3}_{+}$
contains points of $\R^{3}_{++}$ (the open orthant is dense in the closed one), where
$U^{i}_{\lambda}>0=U^{i}_{\lambda}(x)$.
\end{proof}

\begin{corollary}[Concave linearly homogeneous representations]\label{cor:homothetic}
For each explicitly constructed consumer, $U^i_\lambda$ is positively homogeneous of
degree one and concave on $\R^3_+$. Hence the induced preference is homothetic, and its
unique Marshallian demand is linear in wealth. On $\R^3_{++}$ the same preference is
represented by the smooth strictly concave function $u^i_\lambda=\log U^i_\lambda$.
\end{corollary}

\begin{proof}
The inverse-demand map $G=-\nabla\psi$ is homogeneous of degree $-1$, so its inverse
satisfies $Q(tx)=t^{-1}Q(x)$ for $t>0$. By log-homogeneity of the gauge,
\[
u^i_\lambda(tx)=\psi^i_\lambda\bigl(t^{-1}Q(x)\bigr)
=u^i_\lambda(x)+\log t,
\]
and therefore $U^i_\lambda(tx)=tU^i_\lambda(x)$ on the open orthant. Continuity and the
zero boundary extension give the same identity on $\R^3_+$, including $t=0$.

It remains to prove concavity. If $U(x),U(y)>0$, quasiconcavity and homogeneity, applied to
$x/U(x)$ and $y/U(y)$, give $U(x+y)\ge U(x)+U(y)$. If one value is zero, the same
superadditivity follows from monotonicity, because $x+y\ge x$ and $U(y)=0$; if both vanish,
it follows from nonnegativity. Thus $U$ is superadditive. Positive degree-one homogeneity
then yields, for $a\in[0,1]$,
\[
U(ax+(1-a)y)\ge U(ax)+U((1-a)y)=aU(x)+(1-a)U(y),
\]
so $U$ is concave. Wealth linearity of unique Marshallian demand follows from homotheticity
and homogeneity, and is also explicit in \eqref{eq:demand}.
\end{proof}

\begin{lemma}[Demand localization]\label{lem:localization}
Let $k$ be as in Definition~\ref{def:k}, $i\in\{1,2,3\}$, $\lambda\in\Lambda_{1}$, and let
$p\gg0$, $m>0$. On the compact budget set $B(p,m)=\{x\in\R^{3}_{+}:p\cdot x\le m\}$ the
utility $U^{i}_{\lambda}$ attains its maximum; no maximizer lies on
$\partial\R^{3}_{+}$; and the unique maximizer is the interior demand of
Theorem~\ref{thm:rationalize},
\[
x^{i}_{\lambda}(p,m)=-\nabla\psi^{i}_{\lambda}\!\Bigl(\frac{p}{m}\Bigr).
\]
In particular the demand, its exact budget balance, its strict positivity, its degree-zero
homogeneity and its joint smoothness are unchanged by the extension.
\end{lemma}

\begin{proof}
$B(p,m)$ is compact ($p\gg0$) and $U^{i}_{\lambda}$ is continuous on it
(Proposition~\ref{prop:extension}(i)), so a maximizer exists. The bundle
$\varepsilon\one$ with $\varepsilon=m/(\one\cdot p)>0$ is feasible and interior, with
$U^{i}_{\lambda}(\varepsilon\one)>0$; since $U^{i}_{\lambda}\equiv0$ on the boundary, no
boundary point can be optimal. An interior maximizer of $U^{i}_{\lambda}$ on $B(p,m)$
maximizes $U^{i}_{\lambda}$, hence $u^{i}_{\lambda}=\log U^{i}_{\lambda}$, over
$B(p,m)\cap\R^{3}_{++}$; by Theorem~\ref{thm:rationalize} that problem has the unique
solution $x^{i}_{\lambda}(p,m)$, which is therefore the unique maximizer on all of
$B(p,m)$.
\end{proof}

\subsection{The varying-block aggregate}\label{sec:block1}

The individual-demand construction is now complete. We aggregate the three explicit
consumers to determine exactly how their excess demand compares with the target formal
field.

\begin{definition}[Varying block]\label{def:block1}
Consumers $i=1,2,3$ have utilities $U^{i}_{\lambda}$ on $X=\R^{3}_{+}$
(Proposition~\ref{prop:extension}) and endowments
\[
\omega^{i}\deq k\,\mathbf e_i\in X\qquad(\text{$k$ units of good $i$; independent of
$\lambda$}),
\]
which belong to the consumption set, with $\sum_{i=1}^{3}\omega^{i}=k\one\gg0$. Consumer $i$'s
wealth at prices $p\gg0$ is $m_{i}(p)=p\cdot\omega^{i}=k\,p_{i}>0$, the demand is the
unique maximizer on the compact budget set (Lemma~\ref{lem:localization}), and the
individual excess demand is $z^{i}_{\lambda}(p)=x^{i}_{\lambda}\bigl(p,kp_{i}\bigr)-k\,\mathbf e_i$.
\end{definition}

\begin{proposition}[Block-1 aggregation]\label{prop:block1}
For all $p\in\R^{3}_{++}$ and $\lambda\in\Lambda_{1}$,
\begin{equation}\label{eq:block1general}
\sum_{i=1}^{3}z^{i}_{\lambda}(p)
\;=\;k\,R(p)\;-\;\sum_{i=1}^{3}p_{i}\,\nabla\ph^{\,i}_{\lambda}(p),
\qquad
R(p)\deq\frac{\one\cdot p}{3}\,\widehat p^{\,-1}\one-\one.
\end{equation}
Moreover, for all $p$ in the target cone $\Gamma$,
\begin{equation}\label{eq:block1exact}
\sum_{i=1}^{3}z^{i}_{\lambda}(p)\;=\;Z(p;\lambda)\;+\;k\,R(p).
\end{equation}
\end{proposition}

\begin{proof}
By \eqref{eq:demand} with $m=kp_{i}$,
\[
z^{i}_{\lambda}(p)
=\frac{kp_{i}}{3}\widehat p^{\,-1}\one-p_{i}\nabla\ph^{\,i}_{\lambda}(p)-k\,\mathbf e_i.
\]
Summing over $i$ and using $\sum_{i}p_{i}=\one\cdot p$, $\sum_i\mathbf e_i=\one$ gives
\eqref{eq:block1general}. On $\Gamma$, Lemma~\ref{lem:extension}(ii) allows replacing
$\nabla\ph^{\,i}_{\lambda}$ by $\nabla\widehat\varphi^{\,i}_{\lambda}$ (this uses
precisely the flatness of $\chi$ on a \emph{neighborhood} of $\Gamma$: the identity needed
is between \emph{gradients}, not values), and then \eqref{eq:minusZ} gives
$\sum_{i}p_{i}\nabla\widehat\varphi^{\,i}_{\lambda}=-Z(\,\cdot\,;\lambda)$, whence
\eqref{eq:block1exact}.
\end{proof}

Equation~\eqref{eq:block1exact} is the payoff of the explicit block: all parameter
dependence is already exactly correct, and the only remaining discrepancy is the universal
field $kR$, which is independent of $\lambda$.

The field $R$ is homogeneous of degree zero, independent of $\lambda$, and Walrasian,
since $p\cdot R(p)=0$. It is Mantel's universal residual for own-good endowments
$\omega^{i}=k \mathbf e_i$.

\subsection{The residual block via Debreu}\label{sec:block2}

It remains to offset the term $kR$ in \eqref{eq:block1exact}. Since $R$ is independent of
$\lambda$, the residual block can be constructed once and for all by a single application
of Debreu's realization theorem; no parametric realization is required. Although explicit
utility representations are available in related constructions due to Mantel
\cite{mantel74} and Geanakoplos \cite{geanakoplos84}, they are unnecessary here: the
residual consumers are fixed across the parameter family and are used only to generate the
aggregate excess demand $-kR$.

\begin{theorem}[Debreu's exact realization in the normalization used here]\label{thm:debreu}
Let $f\colon\Delta^{\circ}\to\R^{L}$ be continuous on the open normalized price simplex
$\Delta^{\circ}=\{p\in\R^{L}_{++}:\sum_jp_j=1\}$ and satisfy Walras' law
$p\cdot f(p)=0$. For every compact set $P\Subset\Delta^{\circ}$, there exist $L$
consumers with consumption sets $\R^{L}_{+}$, continuous monotone strictly convex
preferences and endowments in $\R^{L}_{+}$ whose individual excess demands
$f^{1},\ldots,f^{L}$ satisfy
\[
\sum_{i=1}^{L} f^{i}(p)=f(p)\qquad(p\in P).
\]
\end{theorem}

\begin{proof}
Debreu \cite{debreu74} states his theorem on the strictly positive part of the unit
Euclidean sphere, rather than on the simplex normalization used here. Extend $f$ from
$\Delta^{\circ}$ to $\R^L_{++}$ by degree-zero rescaling,
\[
\widetilde f(q)\deq f\!\left(\frac{q}{\sum_j q_j}\right).
\]
Then $\widetilde f$ is continuous, homogeneous of degree zero, and Walrasian. Because
$P\Subset\Delta^{\circ}$, its radial image on the positive unit sphere is contained in
a Debreu truncated region $S_{\varepsilon}=\{q:\|q\|=1,\ q_j\ge\varepsilon\ \text{for all }j\}$
for some $\varepsilon>0$. Debreu's exact realization theorem therefore supplies $L$
consumers whose aggregate excess demand equals $\widetilde f$ on $S_{\varepsilon}$.
Individual excess demands are homogeneous of degree zero, so radial rescaling transfers the
identity back to the corresponding normalized simplex region, in particular to $P$. This
is the simplex formulation stated above; see also \cite[Section~17.E]{mwg95}.
\end{proof}

\begin{corollary}[Residual block]\label{cor:residual}
There exist three consumers $(\succeq^{4},\omega^{4})$, $(\succeq^{5},\omega^{5})$,
$(\succeq^{6},\omega^{6})$, all independent of $\lambda$, with continuous, monotone,
strictly convex preferences, such that
\[
\sum_{i=4}^{6}z^{i}(p)\;=\;-\,k\,R(p)\qquad\text{for all }p\in\Gamma.
\]
\end{corollary}

\begin{proof}
\emph{Normalization.} The simplex image
\[
\Delta(\Gamma)=\Bigl\{\tfrac{p}{\one\cdot p}:p\in\Gamma\Bigr\}
\]
is compactly contained in $\Delta^{\circ}$ because $\rho_{1}<1$. Hence
$\Delta(\Gamma)\subset\Delta_{\varepsilon}$ for some $\varepsilon>0$; fix any such
$\varepsilon$. For the illustrative choice $\rho_{1}=0.72$ in Remark~\ref{rem:numbers}, one may take
$\varepsilon=1/20$, since every normalized price
component on $\Delta(\Gamma)$ is at least $0.28/4.44>1/20$.

\emph{Application of Debreu.} The function $f\deq-kR$ is smooth, Walrasian, and independent of $\lambda$.
Theorem~\ref{thm:debreu} with $L=3$ and $P=\Delta_{\varepsilon}$ produces three consumers
with the stated preference properties and $\sum_{i=4}^{6}z^{i}=-kR$ on
$\Delta_{\varepsilon}$.

\emph{Conical extension of the identity.} Each $z^{i}$ ($i=4,5,6$) is the excess demand of
a consumer with wealth $p\cdot\omega^{i}$, hence homogeneous of degree zero in $p$; $R$ is
hom-$0$ as well. Two hom-$0$ functions that agree on $\Delta_{\varepsilon}$ agree on its cone, which
contains $\Gamma$.
\end{proof}

\subsection{Assembly: proof of the realization theorem}\label{sec:assembly}

All ingredients are now available. The theorem follows by combining the explicit
parameterized block with the fixed Debreu residual block and checking the regularity and
continuity properties established above; no new construction is introduced at this stage.

\begin{proof}[Proof of Theorem~\ref{thm:parametric}]
Take $k$ from Definition~\ref{def:k}; consumers $1,2,3$ as in
Definition~\ref{def:block1}; consumers $4,5,6$ from Corollary~\ref{cor:residual}.

\emph{Class membership.} All six consumers have consumption set $X=\R^{3}_{+}$: for
$i=1,2,3$ the utilities $U^{i}_{\lambda}$ are continuous, monotone and quasiconcave on $X$
(Proposition~\ref{prop:extension}(i)--(iii)), hence induce continuous, monotone, convex
preferences; consumers $4,5,6$ carry Debreu's continuous, monotone, strictly convex
preferences on $X$ natively (Corollary~\ref{cor:residual}). Every endowment lies in $X$,
and $\sum_{i=1}^{3}\omega^{i}=k\one\gg0$ already meets the aggregate-resources requirement,
whatever the (nonnegative) residual-block endowments. Hence each $\eco_{\lambda}$ belongs
to the class of Definition~\ref{def:Meconomy}.

\emph{(a).} The global continuity and monotonicity properties and the joint continuity of
$(x,\lambda)\mapsto U^{3}_{\lambda}(x)$ on $\R^{3}_{+}\times\Lambda$ are
Proposition~\ref{prop:extension}(i)--(iii); degree-one homogeneity and concavity are
Corollary~\ref{cor:homothetic}, and vanishing exactly on the boundary follows from the lower
bound in \eqref{eq:sandwich}. The interior strict and smooth properties are
Proposition~\ref{prop:extension}(iv) together with Theorem~\ref{thm:direct};
$U^{1},U^{2}$ are $\lambda$-independent by construction. Demands exist, are
unique and interior at all $p\gg0$, $m>0$ (Lemma~\ref{lem:localization}), which is the
sense in which strict quasiconcavity, and hence strict convexity of the induced preferences,
is needed---and available---only where demand lives.

\emph{(b).} Corollary~\ref{cor:residual}.

\emph{(c).} By Lemma~\ref{lem:localization} the block-1 demands coincide with the interior
demands of Theorem~\ref{thm:rationalize}, so for $p\in\Gamma$ and $\lambda\in\Lambda$,
Proposition~\ref{prop:block1} and Corollary~\ref{cor:residual} give
\[
\sum_{i=1}^{6}z^{i}_{\lambda}(p)
=\bigl(Z(p;\lambda)+kR(p)\bigr)+\bigl(-kR(p)\bigr)
=Z(p;\lambda).
\]
Off the flat region of the cutoff the gauge components are modified, and the identity is
not asserted there.

\emph{(d).} Joint smoothness of the block-1 demands is Theorem~\ref{thm:rationalize}
(unaffected by the extension, by Lemma~\ref{lem:localization}). For the map
$\lambda\mapsto u^{3}_{\lambda}$: by Theorem~\ref{thm:direct}(a) the function
$(x,\lambda)\mapsto u^{3}_{\lambda}(x)$ is $C^{\infty}$ on the open set
$\R^{3}_{++}\times W$ containing $\R^{3}_{++}\times\Lambda$; hence all mixed derivatives
$\partial^{\beta}_{\lambda}\partial^{\alpha}_{x}u^{3}_{\lambda}(x)$ exist and are jointly
continuous, uniformly on compact subsets of $\R^{3}_{++}\times\Lambda$. This is the
primary formulation of (d), and it is what Section~\ref{sec:main} uses. The reformulation as
smoothness of the map $\lambda\mapsto u^{3}_{\lambda}$ from $W$ into the Fr\'echet space
$C^{\infty}(\R^{3}_{++})$ with the compact-open $C^{\infty}$ topology is the exponential
law for spaces of smooth maps \cite[Ch.~I]{km97}; we record it for convenience and rely on the
primary formulation. Joint continuity of $(x,\lambda)\mapsto U^{3}_{\lambda}(x)$ on the
closed orthant is Proposition~\ref{prop:extension}(i). Finally, $k$ and $\chi$ were fixed
once in Assumption~\ref{ass:data} and Definition~\ref{def:k}, the endowments in
Definition~\ref{def:block1} and Corollary~\ref{cor:residual}, and $U^{1},U^{2}$ and the
residual preferences do not involve $\lambda$: the parameter enters through consumer $3$
alone.
\end{proof}

\begin{remark}[Analytic realization versus monodromy topology]\label{rem:interface}
Appendix~\ref{app:analytic} proves only the analytic realization statement. The annular
topology and the branch permutation play no role in constructing the consumers. They enter
only when the exact aggregate identity on the target cone is combined with the formal
covering of Section~\ref{sec:formal}. This modularity is useful: the realization machinery
controls the economic primitives, while the formal model controls the monodromy.
\end{remark}

\section{Examples from the literature}\label{app:anchors}

This appendix applies Theorem~\ref{thm:anchored} to two equilibrium-price configurations
from the literature. In each application we recover the selected price rays and their local
indices from the published primitives and use only those data to specify the initial fibre.
The monodromic family produced by Theorem~\ref{thm:anchored} is newly constructed; no
monodromy property is attributed to the source economy. In particular, the Kehoe source
model contains production whereas the realizing family constructed here is pure exchange.

\subsection{A four-good pure-exchange configuration: Gauthier--Kehoe--Quintin}

Gauthier, Kehoe, and Quintin construct a four-consumer, four-good pure-exchange CES economy
and report fifteen regular equilibria \cite{gkq22}. Their elasticity is $\eta=1/5$. With
$\theta^{i}_{j}=(a^{i}_{j})^{\eta}$, the transformed CES weight matrix $\Theta$ and
endowment matrix $\Omega$ (rows are consumers, columns are goods) are
\[
\Theta=
\begin{pmatrix}
16&4&1&1\\
4&16&1&1\\
1&1&16&4\\
1&1&4&16
\end{pmatrix},
\qquad
\Omega=
\begin{pmatrix}
144&12&1&1\\
12&144&1&1\\
1&1&144&12\\
1&1&12&144
\end{pmatrix}.
\]
Individual demand is
\[
x^{i}_{j}(p)=
\frac{\theta^{i}_{j}\sum_{\ell}p_{\ell}\omega^{i}_{\ell}}
{p_{j}^{\eta}\sum_{\ell}\theta^{i}_{\ell}p_{\ell}^{1-\eta}}.
\]
From their Table~3, choose the symmetric equilibrium $q_{1}$ and the pair $q_{6},q_{7}$:
\begin{gather*}
q_{1}=(0.25,0.25,0.25,0.25),\qquad \ind(q_{1})=-1,\\
q_{6}\approx(0.2564013031,0.2564013031,0.4203610317,0.0668363621),
\qquad \ind(q_{6})=+1,\\
q_{7}\approx(0.2564013031,0.2564013031,0.0668363621,0.4203610317),
\qquad \ind(q_{7})=+1.
\end{gather*}
The last two are exchanged by interchanging goods $3$ and $4$. A high-precision
replication of the CES demand system gives excess-demand residuals below $10^{-45}$.
Our recomputation, in the convention $\ind=\sgn\det(-D_xF)$ of
Section~\ref{sec:prelim}, agrees with the indices reported in their Table~3.
In the orthonormal version of the log-price coordinates used in the proof of
Theorem~\ref{thm:anchored}, with $q_{6}\mapsto-1$ and $q_{7}\mapsto1$, the fixed point is
numerically
\[
z_{0}\approx0.1541033243+0.2179350112i,
\qquad |z_{0}|\approx0.2669147873<1,
\]
so it is uniformly separated from the mobile unit circle.

\begin{corollary}[A monodromic family with the GKQ equilibrium-price fibre]\label{cor:gkq}
There exists an eight-consumer, four-good pure-exchange family with fixed endowments whose
base restricted fibre is exactly $\{q_{6},q_{7},q_{1}\}$ and whose monodromy is the
transposition $(q_{6}\ q_{7})$, with $q_{1}$ fixed.
\end{corollary}

\begin{proof}
Apply Theorem~\ref{thm:anchored} with
$(q_{+},q_{-},q_{0})=(q_{6},q_{7},q_{1})$.
\end{proof}

\subsection{A four-good production configuration: Kehoe}

Kehoe's numerical example has four goods, four Cobb--Douglas consumers, and an
activity-analysis production technology with two non-disposal production activities
\cite{kehoe85num}. The budget-share matrix $\boldsymbol{\alpha}=(\alpha_{ji})$ and the
net-output matrix $\mathcal B$ of the two activities are
\[
\boldsymbol{\alpha}=
\begin{pmatrix}
0.52&0.86&0.50&0.06\\
0.40&0.10&0.20&0.25\\
0.04&0.02&0.2975&0.0025\\
0.04&0.02&0.0025&0.6875
\end{pmatrix},
\qquad
\mathcal B=
\begin{pmatrix}
6&-1\\-1&3\\-4&-1\\-1&-1
\end{pmatrix}.
\]
Kehoe's original endowments are
\[
50\mathbf e_{1},\qquad 50\mathbf e_{2},\qquad
400\mathbf e_{3},\qquad 400\mathbf e_{4}.
\]
Since a common positive rescaling of all endowments merely rescales aggregate consumer
excess demand and the equilibrium activity levels, without changing the equilibrium price
rays or their index signs, we divide all endowments by $10$ and work with
\[
\omega^{1}=5\mathbf e_{1},\qquad \omega^{2}=5\mathbf e_{2},\qquad
\omega^{3}=40\mathbf e_{3},\qquad \omega^{4}=40\mathbf e_{4}.
\]
Here columns of $\boldsymbol{\alpha}$ index consumers. Both non-disposal production
activities operate at all three equilibria.  Normalizing $p_{2}=1$ and writing $p_{4}=t$, the zero-profit equations
give
\[
p(t)=\left(\frac{13-3t}{10},1,\frac{17-7t}{10},t\right).
\]
Substitution into aggregate consumer excess demand reduces equilibrium exactly to
\[
\frac{(t-1)(7980t^{2}-22923t+11713)}
{170(3t-13)(7t-17)}=0.
\]
Thus the three positive price rays correspond to
\[
t_{-}=\frac{22923-\sqrt{151584969}}{15960},\qquad
 t_{0}=1,\qquad
 t_{+}=\frac{22923+\sqrt{151584969}}{15960}.
\]
Numerically, on the slice $p_{2}=1$,
\begin{align*}
k_{-}&\approx(1.1005448265,1,1.2346045952,0.6648505783),\\
k_{0}&=(1,1,1,1),\\
k_{+}&\approx(0.6376882562,1,0.1546059311,2.2077058127).
\end{align*}
Kehoe's bordered-index formula for activity-analysis production gives a directly
reproducible sign calculation \cite{kehoe80}. Using the corresponding reduced coordinate
choice with commodity $4$ omitted, let
\[
\bar J=D_p z^{\mathrm{cons}}_{1:3,1:3},\qquad
\bar B=\mathcal B_{1:3,:}.
\]
Kehoe's bordered determinant has the sign of
\[
\det\begin{pmatrix}-\bar J&-\bar B\\\bar B^{\top}&0\end{pmatrix}.
\]
The accompanying script uses the algebraically equivalent matrix
$\bigl(\begin{smallmatrix}-\bar J&\bar B\\-\bar B^{\top}&0\end{smallmatrix}\bigr)$,
obtained by reversing both off-diagonal block signs, which leaves the determinant unchanged.
Evaluating at the displayed $p_2=1$ representative of each price ray gives approximately
$20.19086705$, $-32.3$, and $3329.629094$. Their magnitudes depend on the chosen
representative and normalization, whereas the signs are the invariant information needed
here. The index pattern is therefore $(+1,-1,+1)$, in agreement with the reduced convention
$\ind=\sgn\det(-D_xF)$ of Section~\ref{sec:prelim}. With $k_{-}$ mapped to $+1$ and
$k_{+}$ to $-1$ in the active coordinate, the fixed unit-price equilibrium has
\[
z_{0}\approx0.5537657458+0.1092493308i,
\qquad |z_{0}|\approx0.5644394719<1.
\]

\begin{corollary}[A monodromic family with Kehoe's equilibrium-price fibre]\label{cor:kehoe}
There exists an eight-consumer, four-good pure-exchange family with fixed endowments whose
base restricted fibre consists exactly of the three Kehoe price rays and whose monodromy
exchanges the two positive-index rays while fixing the unit-price, negative-index ray.
\end{corollary}

\begin{proof}
Apply Theorem~\ref{thm:anchored} to $(k_{-},k_{+},k_{0})$.
\end{proof}

\section{Two-good examples: multiplicity without monodromy}\label{app:twogood}

Section~\ref{sec:twogood} uses the quadratic Toda--Walsh example to distinguish real
economic monodromy from complex algebraic monodromy. For completeness, we record here the
primitive specifications and exact factorizations for the benchmark cases used in their
two-good analysis \cite{todawalsh17}. These calculations add no new topological mechanism:
in every case the positive roots remain globally ordered on the real regular region.

\subsection{Quadratic benchmark}
Toda and Walsh consider two consumers with identical endowments
$e_1=e_2=(e,e)$ and symmetric separable utilities
\[
U_1(x_1,x_2)=\alpha u(x_1)+(1-\alpha)u(x_2),\qquad
U_2(x_1,x_2)=(1-\alpha)u(x_1)+\alpha u(x_2),
\]
where
\[
u(x)=x-\frac{x^2}{2\tau},\qquad \tau>e.
\]
Thus $\alpha$ is a preference weight, not an endowment parameter. After cancellation of
a positive factor, their reduced equilibrium equation can be written as
\[
f_\alpha(p)=
\frac{p(p-1)\bigl[2\alpha(1-\alpha)(p+1)^2-p\bigr]}
{(1-\alpha+\alpha p^2)(\alpha+(1-\alpha)p^2)}=0.
\]
Writing $k=2\alpha(1-\alpha)$, the two non-unit roots solve
\[
kp^2+(2k-1)p+k=0,
\]
hence
\[
p_{\pm}(\alpha)=\frac{1-2k\pm\sqrt{1-4k}}{2k},
\qquad p_-p_+=1.
\]
The condition for three distinct positive real equilibria is
$1-8\alpha(1-\alpha)>0$. The parameters $e$ and $\tau$ do not enter this reduced
equilibrium condition.

\subsection{CRRA benchmark}
Toda and Walsh consider the symmetric CRRA utilities
\[
U_1(x_1,x_2)=\frac{\alpha^{\gamma}x_1^{1-\gamma}+(1-\alpha)^{\gamma}x_2^{1-\gamma}}{1-\gamma},
\qquad
U_2(x_1,x_2)=\frac{(1-\alpha)^{\gamma}x_1^{1-\gamma}+\alpha^{\gamma}x_2^{1-\gamma}}{1-\gamma},
\]
with symmetric endowments $e_1=(e,1-e)$ and $e_2=(1-e,e)$. Their concrete example takes
$\gamma=3$, $\alpha=1/7$, and $e=1/49$ \cite{todawalsh17}. With $t=p^{1/3}$,
\[
z_{1}(t^{3})=
\frac{6t^{2}(t-2)(t-1)(2t-1)}
{7(t^{2}+6)(6t^{2}+1)}.
\]
Hence the positive-price roots are
\[
p\in\left\{\frac18,1,8\right\}.
\]

\subsection{Quasi-linear benchmark}
Their quasi-linear specification is
\[
U_1(x_1,x_2)=x_1+u(x_2),\qquad U_2(x_1,x_2)=u(x_1)+x_2,
\]
with identical endowments $e_1=e_2=(e,e)$. For the concrete example,
$u(x)=\beta x^{1-\gamma}/(1-\gamma)$ with $\beta=1$, $\gamma=3$ (so
$\varepsilon=1/3$), and $e=1/4$ \cite{todawalsh17}. Their quasi-linear formulation allows
the quasi-linear good to take negative values; we use this example only as an algebraic
equilibrium-price benchmark. Again with $t=p^{1/3}$,
\[
z_{1}(t^{3})=\frac14(t-1)(t^{2}-3t+1),
\]
whose positive roots are
\[
p\in\left\{
\left(\frac{3-\sqrt5}{2}\right)^{3},\ 1,\
\left(\frac{3+\sqrt5}{2}\right)^{3}
\right\}.
\]

\subsection{HARA benchmark}
For the HARA case, Toda and Walsh write the Bernoulli utility $u$ through
\[
-\frac{u''(x)}{u'(x)}=\frac{1}{\varepsilon_{\mathrm{TW}}(x-c)}
\]
and use the symmetric preferences
\[
U_1=\alpha^{1/\varepsilon_{\mathrm{TW}}}u(x_1)
 +(1-\alpha)^{1/\varepsilon_{\mathrm{TW}}}u(x_2),\qquad
U_2=(1-\alpha)^{1/\varepsilon_{\mathrm{TW}}}u(x_1)
 +\alpha^{1/\varepsilon_{\mathrm{TW}}}u(x_2).
\]
The concrete benchmark used here has identical endowments $(e,e)$ and parameters
$e=1$, $c=2$, $\varepsilon_{\mathrm{TW}}=-1/4$, and $\alpha=1/4$
\cite{todawalsh17}. With $t=p^{1/4}$,
\[
z_{1}(t^{4})=-\frac{2(t-1)(t+1)}{(t^{5}+3)(3t^{5}+1)}
\bigl(3t^{4}-5t^{3}+3t^{2}-5t+3\bigr).
\]
The palindromic quartic reduces under $y=t+t^{-1}$ to
$3y^{2}-5y-3=0$, yielding one reciprocal pair of positive roots in addition to $p=1$.
Thus all three displayed benchmarks have three distinct simple positive roots
$p_{-}<1<p_{+}$ with $p_{-}p_{+}=1$.

\end{document}